\documentclass[sigconf,balance=false]{acmart}
\usepackage{popets}
\setcopyright{popets}
\copyrightyear{YYYY}
\acmYear{YYYY}
\acmVolume{YYYY}
\acmNumber{X}
\acmDOI{XXXXXXX.XXXXXXX}
\acmISBN{}
\acmConference{Proceedings on Privacy Enhancing Technologies}
\usepackage{bm} 
\usepackage[normalem]{ulem} 
\usepackage[symbol]{footmisc} 
\usepackage{enumitem} 

\usepackage{mathtools,amsmath}

\DeclareMathOperator*{\argmin}{arg\,min}
\newtheorem{hypothesis}{Hypothesis}[section]
\newenvironment{support}{\par\textsc{Support.} \normalfont\ignorespaces}{\par}
\newtheorem{fact}{Fact}[section]

\usepackage{subcaption,wrapfig}

\usepackage{tabularx}
\usepackage{threeparttable}
\usepackage{multirow}
\usepackage{makecell}

\usepackage{algorithmic}
\usepackage[ruled,vlined,linesnumbered]{algorithm2e}

\newcommand{\eldp}{$\varepsilon$-LDP\xspace}

\newcommand{\lossmod}{\mathcal{L}_{\text{mod}}}

\newcommand{\mechanism}{\mathcal{M}}
\newcommand{\indomain}{\mathcal{D}}
\newcommand{\outdomain}{\tilde{\mathcal{D}}}
\newcommand{\loss}{\mathcal{L}}
\newcommand{\trans}{\mathcal{T}}

\begin{document}

\title[Optimal Piecewise-based Mechanism under LDP]{Optimal Piecewise-based Mechanism for Collecting Bounded Numerical Data under Local Differential Privacy}

\author{Ye Zheng}
\orcid{0000-0003-0623-9613}
\affiliation{%
  \institution{Rochester Institute of Technology}
  \city{} 
  \state{}
  \country{}}
\email{ye.zheng@mail.rit.edu}

\author{Sumita Mishra}
\orcid{0000-0002-7093-631X}
\affiliation{
  \institution{Rochester Institute of Technology}
  \city{}
  \state{}
  \country{}}
\email{sumita.mishra@rit.edu}

\author{Yidan Hu}
\orcid{0000-0002-9443-8411}
\affiliation{
  \institution{Rochester Institute of Technology}
  \city{}
  \state{}
  \country{}}
\email{yidan.hu@rit.edu}


\begin{abstract}
    Numerical data with bounded domains is a common data type in personal devices, such as wearable sensors.
    While the collection of such data is essential for third-party platforms, it raises significant privacy concerns.
    Local differential privacy (LDP) has been shown as a framework providing provable individual privacy,
    even when the third-party platform is untrusted. 
    For numerical data with bounded domains, existing state-of-the-art LDP mechanisms are piecewise-based mechanisms, which are not optimal, leading to reduced data utility.
    
    This paper investigates the optimal design of piecewise-based mechanisms to maximize data utility under LDP.
    We demonstrate that existing piecewise-based mechanisms are heuristic instances of the $3$-piecewise mechanism,
    which is far from enough to study optimality. 
    We generalize the $3$-piecewise mechanism to its most general form, i.e. $m$-piecewise mechanism with no pre-defined form of each piece. 
    Under this form, we derive the closed-form optimal mechanism by combining analytical proofs and off-the-shelf optimization solvers.
    Next, we extend the generalized piecewise-based mechanism to the circular domain (along with the classical domain), defined on a cyclic range
    where the distance between the two endpoints is zero. By incorporating this property, we design the optimal mechanism for the
    circular domain, achieving significantly improved data utility compared with existing mechanisms.
    
    Our proposed mechanisms guarantee optimal data utility under LDP among all generalized piecewise-based mechanisms.
    We show that they also achieve optimal data utility in two common applications of LDP: distribution estimation and mean estimation.
    Theoretical analyses and experimental evaluations prove and validate the data utility advantages of our proposed mechanisms.
\end{abstract}

\keywords{local differential privacy, numerical data privacy, bounded domain, circular data}


\maketitle

\section{Introduction}

Numerical data with bounded domains is a fundamental data type in personal devices.
These bounded domains can be categorized into two types: 
linear ranges, such as sensor readings in $[0,1)$, referred to as \emph{classical domain}; 
and cyclic ranges, such as angular measurements in $[0,2\pi)$, referred to as \emph{circular domain}.
These types of data are crucial for third-party platforms to provide personalized services.
However, collecting such data involves particular privacy concerns,
as third-party collectors are potentially untrusted and the data often contain sensitive information. 
Simple anonymization techniques~\cite{DBLP:conf/infocom/YangFX13,DBLP:conf/infocom/NiuLZCL14} 
have been proven insufficient to prevent privacy leakage~\cite{enwiki:1187739761,DBLP:conf/sp/NarayananS08,DBLP:conf/sp/YuFGXZ07}.
Therefore, a provable privacy guarantee is necessary when collecting these sensitive data.

\emph{Local differential privacy} (LDP) serves as a de facto standard,
providing an input-independent formal guarantee regarding the difficulty of inferring sensitive data.
Through the LDP mechanism, sensitive data are randomly perturbed before being sent to an untrusted collector.
The randomization ensures a sufficient level of indistinguishability (indicated by the privacy parameter $\varepsilon$).
Consequently, any observation over the randomized data is essentially powerless to infer the sensitive data.
Laplace Mechanism~\cite{DBLP:conf/tcc/DworkMNS06} is a classical LDP mechanism for numerical data privacy.
It adds random noise, drawn from a Laplace distribution determined by $\varepsilon$, to the sensitive data.
However, the unbounded noise of the Laplace mechanism makes it unsuitable for bounded domains.

State-of-the-art LDP mechanisms for numerical data with bound\-ed domains are \emph{piecewise-based mechanisms}~\cite{DBLP:conf/icde/WangXYZHSS019,DBLP:conf/sigmod/Li0LLS20,ma10430379}.
They are widely used as building blocks
to provide provable privacy guarantees in various scenarios, such as in sensors networks and federated learning.
Piecewise-based mechanisms randomize the sensitive data to a value sampled from a carefully designed piecewise probability distribution.
Existing instantiations use different pieces and probabilities,
but all are designed for classical domains.
Their applicability to other types of bounded domains, e.g. circular domain of angular sensors that commonly appear in personal devices,
is unexplored.

\emph{Data utility} is the most crucial metric for LDP mechanisms, typically measured by the distance between the randomized data and the sensitive data.
While the privacy level is theoretically guaranteed by the privacy parameter $\varepsilon$,
a mechanism with better data utility allows for more accurate analysis. 
The data utility of a piecewise-based mechanism is determined by the distance metric, pieces, and their respective probabilities.
Unfortunately, none of the existing instantiations of piecewise-based mechanisms
are shown to be optimal in terms of data utility, 
indicating potential for improving analysis accuracy without compromising privacy. 
These situations highlight the need for optimal piecewise-based mechanisms.

The optimality of piecewise-based mechanisms remains a challenging problem.
We will see that the existing instantiations are heuristic forms of the $3$-piecewise mechanism (TPM).
As a special case with $3$ pieces and pre-defined forms of pieces,
TPM is far from enough to study the optimality of piecewise-based mechanisms.
\ From the evidence of the staircase Laplace mechanism~\cite{DBLP:conf/isit/GengV14} for unbounded numerical data,
the asymptotically optimal mechanism has a staircase (multiple pieces) form.
For categorical data, a staircase Randomized Response mechanism (SRR)~\cite{DBLP:conf/ccs/WangH0QH22} improves data utility in location collection compared to the general RR mechanism.
Numerous indicators suggest that increasing the variety of probabilities in the data domain, i.e. using more pieces, can improve data utility.
\ In light of these examples, a fundamental question for piecewise-based mechanisms is:
\emph{what is the optimal instantiation of piecewise-based mechanism?}
In the design of piecewise-based mechanisms, the number of pieces, their probabilities and sizes can be arbitrary.
Finding the optimal instantiation within such a large design space is challenging,
as it requires optimizing the number of pieces, their probabilities and sizes simultaneously.

This paper studies the optimality of piecewise-based mechan\-ism in its most general form.\footnote{
    This means that we consider all possible forms of piecewise distributions on a bounded domain, 
    ensuring the most comprehensive generalization.
}
We extend TPM into a generalized piecewise-based mechanism (GPM) that is $m$-piece, 
with each piece having no pre-defined form.
Under GPM, we formulate an optimization problem to minimize the distance between the sensitive and 
randomized data.
By combining the solutions of the optimization problem with analytical proofs, 
we derive the closed-form optimal GPM for classical domains.
\ For circular domains, where distance metrics have a unique property (e.g. the distance between $0$ and $2\pi$ is zero),
we incorporate this property into mechanism design and link the solving of the optimal mechanism to problems in the classical domain.
Table~\ref{tab:existing} summarizes the key features of this paper in comparison with existing works. 
Particularly, our contributions are as follows:
\begin{itemize}
    \item \emph{(Solving framework)} To the best of our knowledge, this is the first work to study the closed-form optimal piecewise-based mechanism under its most general form.
    We propose a framework that integrates analytical proofs with off-the-shelf optimization solvers to derive the closed-form optimal mechanism.
    This approach establishes a feasible foundation for achieving optimal data utility under LDP for numerical data with bounded domains.
    \item \emph{(Closed-form instantiations)} We provide closed-form optimal mechanisms for the classical domain and the circular domain.
    As alternatives to existing mechanisms, they can be directly used as
    building blocks in sensor networks and federated learning, etc,
    while guaranteeing optimal data utility among all piecewise-based mechanisms under LDP.
    \item \emph{(Theoretical and experimental evaluations)} We provide theoretical analyses of data utility and experimental evaluations on 
    two common applications of LDP: distribution and mean estimation.
    The results prove and validate our mechanisms' advantages over existing mechanisms.
    The codes are available at \url{https://github.com/ZhengYeah/Optimal-GPM}.
\end{itemize}

\textbf{Structure.} The main part of this paper is organized as follows:
After the preliminaries, we present the optimal piecewise-based mechanism for the classical domain in Section~\ref{sec:optimal}.
Section~\ref{sec:cyclic_data} focuses on the circular domain and derives its optimal mechanism.
Following this, Section~\ref{sec:estimation} discusses the optimality when applying to two common tasks: distribution and mean estimation.

\section{Preliminaries} \label{sec:preliminaries}

This section formulates the problem and the concept of local differential privacy (LDP).
We present existing instantiations of TPM and their limitations, which motivate our proposed 
optimal generalized piecewise-based mechanism (OGPM).

\subsection{Problem Formulation} 

We consider a typical data collection schema that consists of a set of \emph{users} and one \emph{collector}.
Each user has a numerical sensitive data $x_i \in \indomain$, where $\indomain$ is a continuous and bounded domain.
The collector needs to collect data from users for statistical estimations, such as the mean value and distribution of the data.

However, the collector is untrusted and may attempt to infer users' sensitive data.
To protect privacy, each user locally randomizes their sensitive data using a privacy mechanism $\mechanism: \indomain \to \outdomain$,
then sends $y_i=\mechanism(x_i)$ to the collector.

We seek to design an optimal $\mathcal{M}$ that maximizes the data utility by minimizing the distance between 
the sensitive data $x_i$ and the reported data $y_i$, while ensuring
$\varepsilon$-LDP (Definition~\ref{def:edldp}).

\subsection{Local Differential Privacy}

\begin{definition}[$\varepsilon$-LDP~\cite{DBLP:journals/corr/DuchiWJ16}]  \label{def:edldp}
    A randomization mechanism $\mathcal{M}: \indomain \to \outdomain$ satisfies $\varepsilon$-LDP,
    if for two arbitrary inputs $x_1$ and $x_2$, the probability ratio of outputting the 
    same $y$ is bounded:
    $$
    \forall x_1, x_2 \in \indomain, \forall y \in \outdomain: \frac{\Pr[\mathcal{M}(x_1) = y]}{\Pr [\mathcal{M}(x_2) = y]} \leq \exp(\varepsilon).
    $$
\end{definition} 

If $\mechanism(x)$ is continuous, the probability $\Pr[\cdot]$ is replaced by probability density function ($pdf$).
Intuitively,
\eldp represents the difficulty of distinguishing between $x_1$ and $x_2$ given $y$.
A lower privacy parameter $\varepsilon \in [0, +\infty)$ means higher privacy. 
\ For example, $\varepsilon = 0$ requires that $\mathcal{M}$ maps two arbitrary 
inputs to any output $y$ with the same probability, 
thus $y$ contains no distribution information about $x$, 
making any hypothesis-testing method to infer the sensitive $x$ powerless even with known $\mathcal{M}$.

\subsection{Piecewise-based Mechanisms} \label{subsec:piecewise}

\begin{table}[t]
    \begin{center}
        \caption{OGPM vs existing instantiations of TPM.}\label{tab:existing}
        \resizebox*{\linewidth}{!}{
            \begin{threeparttable}
                \begin{tabular}{lrrrr}
                    \toprule
                    & \textbf{Domain} & \textbf{Optimality} & \textbf{Closed form} & \textbf{Estimation} \\
                    \midrule
                    PM~\cite{DBLP:conf/icde/WangXYZHSS019} & \multirow{3}{*}{Classical} & No & Yes & Mean \\
                    SW~\cite{DBLP:conf/sigmod/Li0LLS20} & & No & Yes & Distribution \\
                    PTT~\cite{ma10430379} & & Partly* & No & Mean \\
                    \midrule
                    \makecell[l]{This paper \\ (OGPM)} & \makecell[r]{Classical \\ \& circular} & Yes & Yes & \makecell[r]{Mean \& \\distribution}\\
                    \bottomrule
                \end{tabular}
                \begin{tablenotes}
                    \item{\footnotesize * Proved the existence of the optimal under TPM, but did not give closed-form instantiations.
                    Appendix~\ref{appendix:ptt} provides detailed discussion.}
                \end{tablenotes}
            \end{threeparttable}
        }
    \end{center}
\end{table}

When the input domain $\indomain$ is both continuous and bounded, 
the state-of-the-art mechanisms to achieve $\varepsilon$-LDP are piecewise-based mechanisms.
We summarize these mechanisms as heuristic instances of the following definition.

\begin{definition} \label{def:piecewise}
    $3$-piecewise mechanism (TPM) $\mathcal{M}: \indomain \to \outdomain_{\varepsilon}$ is a family of probability density functions that,
    given input $x\in \indomain$, outputs $y\in \outdomain$ according to
    \begin{equation*}
        pdf[\mathcal{M}(x) = y] = 
        \begin{cases}
            p_{\varepsilon} & \text{if} \ y \in [l_{x,\varepsilon}, r_{x,\varepsilon}], \\ 
            p_{\varepsilon} / \exp{(\varepsilon)} & \text{if } y \in \outdomain_{\varepsilon} \, \backslash \, [l_{x,\varepsilon}, r_{x,\varepsilon}],
        \end{cases}
    \end{equation*}
    where $p_{\varepsilon}$ is a variable determined solely by $\varepsilon$,\footnote{
        Otherwise, if $p_{\varepsilon}$ varies with $x$,
        it violates the \eldp constraint because the probability ratio outputting the same $y$
        from $x_1$ and $x_2$ is not bounded by $\exp(\varepsilon)$.
    }
    while $l_{x,\varepsilon}$ and $r_{x,\varepsilon}$ depend on $x$ and $\varepsilon$.
    The output domain $\outdomain_{\varepsilon} \supset \indomain$ is an enlarged domain depending on $\varepsilon$.
\end{definition}

TPM samples the output $y$ for each $x$ from a piecewise distribution. 
This sampling is with a higher probability $p_{\varepsilon}$ 
within $[l_{x,\varepsilon},r_{x,\varepsilon}]$
and a lower probability $p_{\varepsilon} / \exp(\varepsilon)$ within the remaining 
two pieces $\outdomain \, \backslash\, [l_{x,\varepsilon},r_{x,\varepsilon}]$, 
satisfying the $\varepsilon$-LDP constraint.

\textbf{Instantiations.}
In TPM, the parameters are the central interval $[l_{x,\varepsilon}, r_{x,\varepsilon}]$, its probability $p_{\varepsilon}$, 
and the output domain $\outdomain_{\varepsilon}$.
Different instantiations of those parameters yield different existing 
mechanisms~\cite{DBLP:conf/icde/WangXYZHSS019,DBLP:conf/sigmod/Li0LLS20,ma10430379}.
For example, PM~\cite{DBLP:conf/icde/WangXYZHSS019} is the first instantiation of TPM,
defined on $[-1, 1] \to [-C_{\varepsilon}, C_{\varepsilon}]$,
where $C_{\varepsilon}$ is a variable determined solely by $\varepsilon$ and 
the central interval has a fixed length $r_{x,\varepsilon} - l_{x,\varepsilon} = C_{\varepsilon} - 1$.

\textbf{Data utility metric.}
To quantify the data utility of different instantiations,
we consider the general $L_p$-similar error metric (i.e. $|y-x|^p$) as a loss function $\mathcal{L}:\mathbb{R} \to \mathbb{R}$.
Thus, the error is:
\begin{equation} \label{equ:utility_1}
    Err(x) = \int_{\outdomain} \mathcal{L}(y, x)\mathcal{P}_{\mathcal{M}(x)} \mathrm{d}y,
\end{equation}
where $\mathcal{P}_{\mathcal{M}(x)}$ is the pdf defined by $\mechanism(x)$.
$Err(x)$ illustrates the expected error when applying $\mechanism$ on $x$ under the loss function $\mathcal{L}$.
For example, $\loss(y, x) \coloneq |y - x|$ is the absolute error, and $\mathcal{L}(y, x) \coloneq (y - x)^2$ is the square error.
Then $Err(x)$ corresponds to the mean absolute error (MAE) and mean square error (MSE)~\cite{DBLP:conf/icde/WangXYZHSS019,ma10430379}, respectively.
Lower $Err(x)$ indicates better data utility.

\textbf{Limitations of TPM.} 
Existing instantiations of TPM have the following limitations.
\begin{itemize}
    \item \textit{Not optimal in data utility.}
    None of the existing instantiations provided closed forms for the optimal data utility. 
    Meanwhile, they also assume an invariable length $r_{x,\varepsilon} - l_{x,\varepsilon}$ of the central piece for all $x$,
    and symmetric probability $p_{\varepsilon} / \exp(\varepsilon)$ for the remaining two pieces.
    However, a general-form piecewise-based mechanism can have more pieces, 
    unfixed piece lengths, and asymmetric probabilities, potentially improving data utility.
    \item \textit{Limited applicability.} 
    Existing instantiations of TPM have enlarged and unfixed output domains $\outdomain_{\varepsilon} \supset \indomain$.
    Enlarged output domain incurs applicability issues in scenarios where the collector requires 
    the output domain to align with the input domain (i.e. $\outdomain = \indomain$),\footnote{
        While post-processing the output by truncating it to $\indomain$ is possible, this approach may still result in low data utility. 
        Sections~\ref{subsec:comparison_with_PM_SW_original} and \ref{subsec:staircase_truncated} provide comparisons with mechanisms that include truncation.
    } 
    such as in common sensor-based services.
    \end{itemize}

\section{Generalized Piecewise-based Mechanism} \label{sec:optimal}

This section generalizes TPM to its most general form (GPM). 
We introduce a framework for deriving the closed-form optimal GPM for the classical domain.

\begin{definition} \label{def:general_piecewise}
    Generalized $m$-piecewise mechanism ($m$-GPM)
    $\mathcal{M}: \indomain \to \outdomain$ is a family of probability density functions
    that, given input $x\in \indomain$, outputs $y\in\outdomain$ according to
    \begin{align*}
        pdf[\mathcal{M}(x) = y] &= 
        \begin{dcases}
            p_{1,\varepsilon} & \text{if} \ y \in [l_{1,x,\varepsilon}, r_{1,x,\varepsilon}), \\ 
            \ \vdots  & \vdots \quad \quad \ \ \vdots \\
            p_{m,\varepsilon} & \text{if} \ y \in [l_{m,x,\varepsilon}, r_{m,x,\varepsilon}), \\
        \end{dcases} \\[0.5em]
        \forall i, j \in [m], &\max \frac{p_{i,\varepsilon}}{p_{j,\varepsilon}} \leq \exp(\varepsilon),
    \end{align*}
    where $[m] \coloneq \{1, ..., m\}$.
    Each probability $p_{i,\varepsilon}$ depends solely on privacy parameter $\varepsilon$,
    while interval boundaries $l_{i,x,\varepsilon}$ and $r_{i,x,\varepsilon}$ depend on both $x$ and $\varepsilon$.
\end{definition}
An $m$-GPM partitions its output domain into $m$ pieces, assigning probability $p_{i,\varepsilon}$ to each piece 
$[l_{i,x,\varepsilon}, r_{i,x,\varepsilon})$. 
The probabilities $p_{i,\varepsilon}$ are independent of the input $x$,
and their ratios must be bounded by $\exp(\varepsilon)$ to satisfy $\varepsilon$-LDP.
For notational clarity, we omit subscripts $x$ and $\varepsilon$ when their context is clear. 
Additionally, $\mathcal{M}$ must satisfy standard probability requirements: non-negativity ($p_i \geq 0$), 
continuity ($r_i = l_{i + 1}$), and normalization. 
TPM is a special case of GPM with $m=3$.


Finding the optimal GPM requires determining both the optimal number of pieces $m$ and the corresponding $p_{i,\varepsilon}, l_{i,x,\varepsilon}, r_{i,x,\varepsilon}$.
Due to the infinite possibilities for $m \in \mathbb{N}^+$ and the resulting $3m$ variables, analytical solutions are computationally intractable.
We therefore propose a framework that combines analytical proofs with off-the-shelf optimization solvers.


\subsection{Framework for Deriving the Optimal GPM} \label{subsec:framework}

To derive the \emph{closed-form} optimal GPM, we
(i) formulate finding the optimal $m$-GPM as an optimization problem;
(ii) determine the optimal $m$ based on the solutions of the optimization problem;
(iii) derive the optimal closed-form expression (among all $m$-GPM).

\textbf{Optimal \bm{$m$}-GPM.}
To find the optimal GPM instantiation with $m$ pieces, we need to determine the variables $p_i$, $l_i$, and $r_i$.
Any feasible assignment of these variables yields a mechanism $\mechanism$ whose utility can be measured by $Err(x)$ from Formula~(\ref{equ:utility_1}).
Finding the optimal $m$-GPM requires solving a min-max optimization problem that minimizes the worst-case error over all possible inputs $x$:\footnote{
    Worst-case error is the most common utility metric in mechanism design~\cite{DBLP:conf/icde/WangXYZHSS019,ma10430379}.
    We can also optimize the error at other specific points, see Section~\ref{sec:optimize_x}.
}
\begin{equation} \label{equ:min-max}
    \begin{gathered}
        \min_{p_i, l_{i,x}, r_{i,x}} \max_{x} \int_{\outdomain} \loss(y, x)\mathcal{P}_{\mechanism(x)}\mathrm{d}y, \\
        \text{s.t.} \ \mechanism \ \text{satisfies Definition~\ref{def:general_piecewise}}. 
    \end{gathered}
\end{equation}
This formulation yields the optimal $x$-independent $p_i$ values and the corresponding $l_{i,x}, r_{i,x}$ for the worst-case input $x$.
However, since these $l_{i,x}, r_{i,x}$ may not be optimal for other inputs, we need a second optimization step using the obtained optimal $p_i$:
\begin{equation} \label{equ:lr_i}
    \begin{gathered}
        \min_{l_{i,x}, r_{i,x}} \int_{\outdomain} \loss(y, x)\mathcal{P}_{\mechanism(x)}\mathrm{d}y, \\
        \text{s.t.} \ \mechanism \ \text{satisfies Definition~\ref{def:general_piecewise} with } p_i. 
    \end{gathered}
\end{equation}
Together, these two steps determine the optimal instantiation of $m$-GPM for any given domain mapping $\indomain \to \outdomain$, distance metric $\loss$, piece number $m$, privacy parameter $\varepsilon$, and input $x$. Figure~\ref{fig:solving_flow} illustrates this solving process.

\begin{figure}[t]
    \centering
    \includegraphics[width=0.85\linewidth]{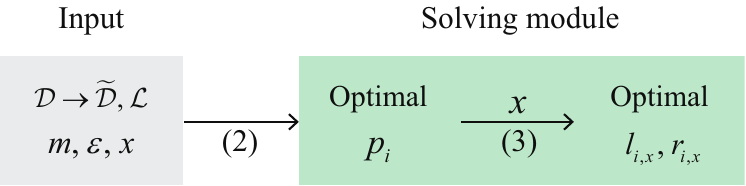}
    \caption{Solving flow for the optimal $m$-GPM.
    Two arrows indicate problems in~(\ref{equ:min-max}) and~(\ref{equ:lr_i}).}
    \label{fig:solving_flow}
\end{figure}

\textbf{Challenges.}
Nonetheless, solving Formulation~(\ref{equ:min-max}) has practical difficulties.
Even if it can be solved,
there is still a gap between the solved optimal $m$-GPM and the closed-form optimal GPM (among all $m$-GPM).
We detail them as follows.
\begin{itemize}
    \item \emph{(Solving difficulty)} Formulation~(\ref{equ:min-max}) is a min-max problem, and the integrand $\loss(y, x)\mathcal{P}_{\mechanism(x)}$ is non-linear.
    It is a non-convex problem whose global optimal hard to solve.
    \item \emph{(Optimal $m$)} The solved optimal is only for $m$-GPM given $m$. It is necessary to find the optimal piece number $m$. 
    \item \emph{(Closed form)} For practical usage, we need closed-form $p_{i}$, $l_{i,x}$ and $r_{i,x}$ (i.e. their relationships with $\varepsilon$ and $x$), 
    rather than their specific values for every $\varepsilon$ and $x$.
\end{itemize}

\textbf{Solutions.} 
We address the challenges and gaps with the following solutions.
\begin{itemize}
    \item Formulation~(\ref{equ:min-max}) can be simplified to two \emph{bilinear optimization} problems that
    can be solved by off-the-shelf solvers. (Section~\ref{subsec:reduced_form})
    \item If the optimal ($m+1$)-GPM is identical to the $m$-GPM,
    then $m$ is the optimal piece number. (Section~\ref{subsec:optimal_piece_number})
    \item Leveraging the above results, the optimal closed-form $p_{i}$, $l_{i,x}$
    and $r_{i,x}$ can be obtained by analytical deduction (for TPM)
    or numerical regression (for any $m$-GPM). (Section~\ref{subsec:closed_form})
\end{itemize}
Following these solutions, we can obtain the closed-form optimal GPM 
(among all $m$) for any $\varepsilon$ given $\indomain \to \outdomain$ and $\loss$.
Before presenting the detailed solutions, we first discuss the conditions under which the obtained GPM is optimal.

\label{anchor:optimality_conditions}
\textbf{Conditions for optimality.}
When discussing optimality, the following aspects should be specified:
(i) the error metric, (ii) the data domain and family of mechanisms, (iii) the strength of the optimality, and (iv) whether post-processing is allowed.
In this paper, the optimality of GPM is defined with respect to:
(i) the worst-case $L_p$-similar error metric, 
(ii) bounded numerical domains $\indomain \to \outdomain$ and mechanisms based on piecewise distributions, 
(iii) minimization of error value (not asymptotic or order-of-magnitude optimality), and 
(iv) without post-processing.
These conditions are widely applicable in practice and literature.
However, varying any of them may lead to different optimality results.
Appendix~\ref{appendix:optimality_discussion} provides a detailed discussion of these conditions and related optimalities.

\subsection{Detailed Solutions} \label{subsec:detailed_answers}

\subsubsection{Solution 1: Simplified Form} \label{subsec:reduced_form}
The min-max problem in Formulation~(\ref{equ:min-max}) can be simplified.
The key observation is that its inner maximization term, $\max_{x}$, has a closed form, i.e. the worst-case error is from the endpoints of $\indomain$.
Lemma~\ref{lemma:endpoint} states this observation.

\begin{lemma} \label{lemma:endpoint}
    Assume $\indomain = [a,b)$, the objective of Formulation~(\ref{equ:min-max}) can be simplified to
    \begin{equation*} 
    \min_{p_i, l_{i,x}, r_{i,x}}\max_{x\in\{a,b\}} \int_{\outdomain} \loss(y, x)\mathcal{P}_{\mechanism(x)}\mathrm{d}y.
    \end{equation*}
\end{lemma}

\begin{proof}
    (Sketch) The key of the proof is that each integral on $\outdomain$ is convex function w.r.t $x$.
    Thus, their non-negative weighted sum is also convex.
    According to the Bauer maximum principle~\cite{noauthor_bauer_2024}, 
    the maximum is achieved at the endpoints of $\indomain$,
    i.e. $x = a$ or $b$.\footnote{
        We use $[a,b)$ to denote the domain for consistency with the form of pieces in GPM.
        This is the same as $[a,b]$ in implementation.
    }
    Appendix~\ref{appendix:endpoint} provides the full proof.
\end{proof}


\textbf{Complexity.} 
(i) Lemma~\ref{lemma:endpoint} simplifies Formulation~(\ref{equ:min-max}) to two \emph{bilinear optimization} problems,
i.e. when $x=a$ and $x=b$ respectively.
The integrand $\loss(y, a)\mathcal{P}_{\mechanism(a)}$ includes terms such as $p_i l_i$ and $p_i r_i$,
which involve multiplications of two variables.
This problem can be solved by off-the-shelf solvers such as Gurobi~\cite{noauthor_gurobi_nodate},
which employs stochastic \emph{Branch and Bound}~\cite{DBLP:journals/jgo/GupteADC17} method to handle the bilinear terms.
\ (ii) The number of variables is at most $3m$, which can be efficiently solved for small $m$.
For example, we can obtain the exact optimal for $m \leq 7$ and $\loss = |y-x|$ within $2$ seconds.
Furthermore, solvers generally provide over- and under-approximation for bilinear problems~\cite{nekst-online_what_2016}.
We can obtain solutions with $\leq1\%$ gap from the optimal for $m \leq 19$ within $1$ minute. 
\ (iii) Formulation~(\ref{equ:lr_i}) for solving $l_i$ and $r_i$ has at most $2m$ variables and without $p_i$ terms,
it is a toy-size problem when $m$ is small.

\textbf{Encoding details.} 
We need to encode the problem in Lemma~\ref{lemma:endpoint} to simple mathematical expressions that can be handled by the solver.
If we instantiate $\loss = |y-x|$ and focus on the left endpoint $x = a$, this problem becomes:
\begin{equation*}
    \begin{split}
        &\min_{p_i, l_{i}, r_{i}} \sum_{i=1}^{m} p_i \int_{l_{i}}^{r_{i}} (y-a) \mathrm{d}y \\
        = &\min_{p_i, l_{i}, r_{i}} \sum_{i=1}^{m}\frac{p_i}{2} \left((r_{i} - a)^2 - (l_{i} - a)^2\right).
    \end{split}
\end{equation*}
This problem is a bilinear optimization problem.
The highest-degree term is $p_i \cdot r_{i} \cdot r_{i}$, which can be reformulated as $p_i \cdot t$ with $t = r_{i} \cdot r_{i}$, i.e. multiplication of bilinear terms.
Such problems can be solved by off-the-shelf bilinear solvers~\cite{noauthor_gurobi_nodate}.
The other problem in Formulation~(\ref{equ:lr_i}) can be encoded similarly but is much easier due to constant $p_i$.

\subsubsection{Solution 2: Optimal Piece Number}\label{subsec:optimal_piece_number}

Although we can obtain the optimal $m$-GPM for any $m$ given sufficient time, solving for each $m$ is unnecessary.
The following lemma provides a theoretical basis for capping the optimal number of pieces.

\begin{lemma} \label{lemma:plus_piece}
    For all possible $\varepsilon$ and $x$, if the optimal ($m+1$)-GPM is the same as the $m$-GPM,
    then the optimal piece number is $m$. 
\end{lemma}

\begin{proof}
    (Sketch) 
    The key insight is that, if the optimal piece number is not $m$, 
    i.e. an additional piece can lower the error, then this additional piece will be captured by the optimal ($m+1$)-GPM.
    \ Therefore, (i) if $m$ is not the optimal piece number, then the optimal ($m+1$)-GPM is different from the optimal $m$-GPM.
    (ii) if $m$ is the optimal piece number, then there is no additional piece can lower the error,
    making the optimal ($m+1$)-GPM is the same as the optimal $m$-GPM.
    \ Now, no additional piece can be captured, 
    hence $m$ is the optimal piece number.
    Appendix~\ref{appendix:plus_piece} provides the full proof. 
\end{proof}

Lemma~\ref{lemma:plus_piece} requires checking the results for every $\varepsilon$ and $x$,
which is challenging as $\varepsilon \in [0, \infty)$ and $x \in \indomain$ are infinite sets.
In practice, we restrict $\varepsilon$ within its generally meaningful domain, e.g. $\varepsilon \in [0, 10)$, and
employ \emph{Monte Carlo random sampling} to generate
a set of random pairs $(\mathcal{E,X}) = \{(\varepsilon_1, x_1), (\varepsilon_2, x_2), \dots, (\varepsilon_n,x_n)\}$.
If Lemma~\ref{lemma:plus_piece} holds for the $n$-size set $(\mathcal{E,X})$, 
the optimality of $m$-GPM is guaranteed with probability $1$ as $n \to \infty$~\cite{DBLP:books/lib/Rubinstein81}.

\begin{figure}[t]
    \centering
    \begin{subfigure}[b]{0.48\linewidth}
        \centering
        \includegraphics[width=0.98\linewidth]{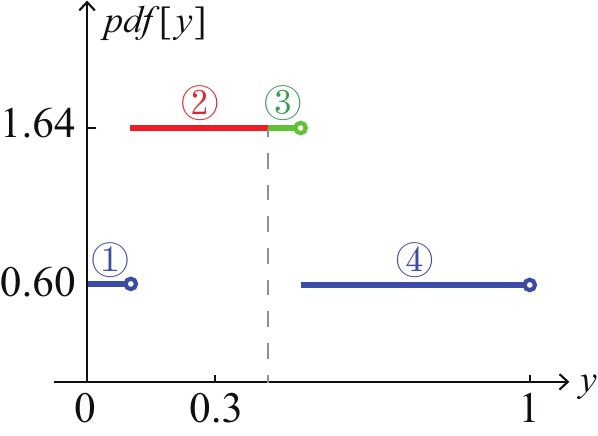}
        \caption{$m = 4$.}
    \end{subfigure}
    \hfill
    \begin{subfigure}[b]{0.48\linewidth}
        \centering
        \includegraphics[width=0.98\linewidth]{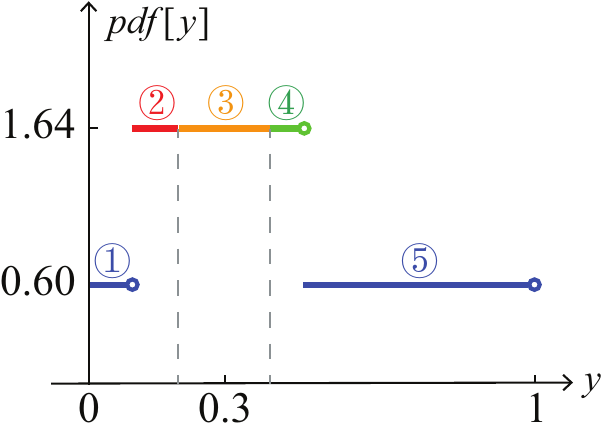}
        \caption{$m = 5$.}
    \end{subfigure}
    \caption{Optimal $4$-GPM and $5$-GPM when $\varepsilon = 1$ and $x = 0.3$.
    They are identical as $m = 3$ after merging redundant pieces.}
    \label{fig:345_piece_example}
\end{figure}

\textbf{Results for Monte Carlo sampling.}
Fixing $\indomain = \outdomain = [0,1)$, for $\loss = |y-x|$ and $\loss = (y-x)^2$,
the optimal $4$-GPM are identical as the optimal $3$-GPM
for random $(\mathcal{E,X})$ with at least $n=10^4$.
These optimal results align with TPM,
but the optimal $p,l$ and $r$ values are different from existing instantiations.

Given the continuity of the objective functions w.r.t $\varepsilon$ and $x$,
we posit that it is, in fact, the exact optimal.
Therefore, when we say the optimal piece number is $m=3$,
it implies the \emph{statistical optimality} guaranteed by the Monte Carlo asymptotic
technique with a strength of at least $n=10^4$ random samples.


\begin{example}
    For $[0,1)\to[0,1)$ and $\loss = |y-x|$,
    Figure~\ref{fig:345_piece_example} shows two examples of the optimal $4$-GPM and optimal $5$-GPM
    when $\varepsilon = 1$ and $x = 0.3$.
    After merging redundant pieces, i.e. connected pieces with the same probability,
    they are the same as the $3$-GPM and fall into the TPM category.
\end{example}

\subsubsection{Solution 3: Closed-form Instantiation} \label{subsec:closed_form}
After determining the optimal $m$, we can derive the closed-form instantiation. 
If the optimal results exhibit $m=3$ and coincide with TPM,
it facilitates analytical deduction for the closed-form optimal $p$, $l$ and $r$.  
Otherwise, numerical regression can be employed to obtain the closed-form instantiation.
Figure~\ref{fig:solving_flow_2} illustrates the workflow.

\textbf{Analytical deduction.}
For TPM, the optimization procedure for solving $p_i$, $l_{i}$ and $r_{i}$
can be conducted analytically.
The key observation is that there are only three variables: $p, l$, and $r$ in TPM.
Due to the normalization constraint of probability, 
the central interval length $r - l$ can be replaced by its probability $p$.
This reduces the solving for the optimal $p$ to a univariate optimization problem w.r.t. $p$,
which can be solved by analyzing the first-order derivative.
With the solved $p$, solving $l$ and $r$ also becomes a univariate optimization problem.
Appendix~\ref{appendix:deduction} provides the formalized process.

\textbf{Numerical regression.}
For any $m$-GPM, we can obtain the closed-form
$p_i$, $l_{i}$, and $r_{i}$ through numerical regression on their solved optimal values.

Assume that we want to find the closed-form optimal $p_i$.
Given $(\mathcal{E},\mathcal{P}_i) = \{(\varepsilon_1, p_{i,1}), \dots, (\varepsilon_n, p_{i,n})\}$, 
which contains the solved optimal $p_i$ for random $\varepsilon$, we aim to find the relationship between $\varepsilon$ and $p_i$.
This relationship can be approximated by $\hat{p}_i = f(\varepsilon,\beta)$, where $f$ is a designed feature with $\beta$ as regression parameters. 

Ideally, the designed feature $f$ matches the truth form of $p_i$.
In this case, the regression result of $f$ on $(\mathcal{E},\mathcal{P}_i)$ converges to the optimal closed-form $p_i$.
If not, the regression result may not converge to the optimal.
In practice, we can use heuristic forms of $f$ for tighter approximation.
Due to the structure of the LDP constraint, we suggest choosing $p_i$ with a form of $\exp(\beta_1\varepsilon)$, 
allowing us to design $f = \exp(\beta_1\varepsilon) + \beta_2$. 

\begin{example}
    For $\loss = |y-x|$, $\indomain = \outdomain = [0, 1)$ and $m = 3$, 
    assume the probability $p$ has the form $\hat{p} = \exp(\beta_1\varepsilon) + \beta_2$.
    We use the \verb|scipy.curve_fit| package to regress $p$ on $50$ random $\varepsilon$;
    then its regression result is $\hat{p} = \exp(\varepsilon/2) - 0.07$
    with a maximal error $\leq 10^{-2}$.
    Note that this result almost coincides with the ground-truth $p$ in Theorem~\ref{theo:optimal_concretization}.
\end{example}

\begin{figure}[t]
    \centering
    \includegraphics[width=\linewidth]{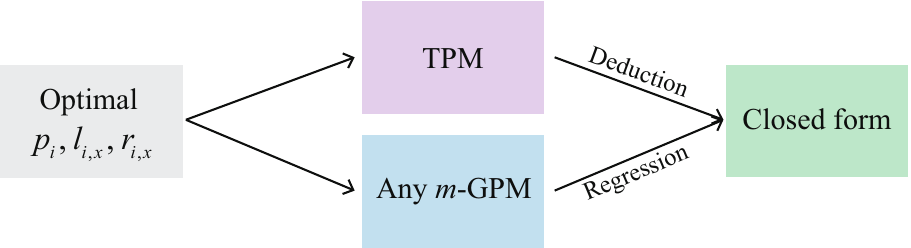}
    \caption{Solving flow for the optimal closed form.}
    \label{fig:solving_flow_2}
\end{figure}

\subsection{Closed Form for Classical Domain} \label{subsec:results_classical}

The above framework for deriving the closed-form optimal GPM 
is applicable to any $\indomain \to \outdomain$.
Using this framework, this subsection provides instantiations for a common case: $\indomain = \outdomain$.

\textbf{Restricting domain.}
In real-world applications, the output domain is often required to match the input domain, i.e. $\indomain=\outdomain$.\footnote{
    $\outdomain$ generally should be larger than or equal to $\indomain$ to ensure the mechanism's meaningfulness;
    otherwise, it means some ranges in $\indomain$ will disappear after applying the mechanism.
}
Furthermore, a concrete domain (e.g. $\indomain =[0,1)$) does not limit the generality,
as $\indomain$ can be transformed to other domains through scaling and shifting operations. 
The theorem below characterizes the privacy and utility invariants under such transformations.

\begin{theorem}[Transformation invariants] \label{theo:transformation_invariant}
    Given a GPM $\mechanism: \indomain \to \outdomain$ satisfying $\varepsilon$-LDP, if domain $\indomain' = c\indomain + d$ and $\outdomain' = c\outdomain + d$ with $c>0$, then transformation
    \begin{equation*}
        \trans(\mechanism): x'=cx+d,\  p'_{i} = p_{i} / c,\  [l'_{i}, r'_{i}) = c[l_{i}, r_{i}) + d
    \end{equation*}
    results in GPM $\mechanism' = \trans(\mechanism): \indomain' \to \outdomain'$ having the same privacy level $\varepsilon$. (privacy invariant)
    
    Meanwhile, if there is another GPM $\mechanism_{\text{bad}}:\indomain \to \outdomain$ with error $Err(x,\mechanism) \leq Err(x, \mechanism_{\text{bad}})$ for a given $x$,
    then 
    \begin{equation*}
        Err(x', \mechanism') \leq Err(x', \mechanism'_{\text{bad}}),
    \end{equation*}
    i.e. $\trans$ maintains the data utility ordering. (utility invariant)
\end{theorem}

\begin{proof}
    (Sketch) The privacy invariant is the same as applying a linear \emph{post-processing} of DP to GPM.
    The utility invariant is due to $\trans$ as a linear function on $\indomain$.
    Appendix~\ref{appendix:transformation_invariant} provides the full proof of these two invariants.
\end{proof}

Theorem~\ref{theo:transformation_invariant} allows us to discuss the optimality on a fixed input domain.
If a mechanism's input domain differs from $\indomain$, 
we can transform it to $\indomain$, and this transformation maintains the optimality.

\begin{hypothesis} \label{theo:optimal_form}
    For any domain $\indomain \to \indomain$, under absolute error and square error metrics,
    the optimal piecewise-based mechanism falls into $3$-GPM.
\end{hypothesis}

\begin{support}
    We validated this hypothesis for $\indomain = [0,1)$ by performing Monte Carlo sampling on $10^4$ random $(\varepsilon, x)$ pairs (detailed in Section~\ref{subsec:optimal_piece_number}). 
    Theorem~\ref{theo:transformation_invariant} then extends this optimality to any $\indomain$.
    Given the continuity of the objective functions w.r.t. $\varepsilon$ and $x$, 
    we posit that $m=3$ is indeed the exact optimal.
    To further support this hypothesis, Appendix~\ref{appendix:one_direction} outlines two directions for analytical proof and highlights the associated challenges.
\end{support}

\begin{theorem} \label{theo:optimal_concretization}
    If hypothesis~\ref{theo:optimal_form} holds, then
    GPM $\mathcal{M}: [0,1) \to [0,1)$ with the following closed-form instantiation
    \begin{equation*}
        pdf[\mathcal{M}(x) = y] = 
        \begin{dcases}
            p_{\varepsilon} & \text{if} \ y \in [l_{x,\varepsilon}, r_{x,\varepsilon}), \cr
            p_{\varepsilon} / \exp{(\varepsilon)} & y \in [0,1) \, \backslash \, [l_{x,\varepsilon}, r_{x,\varepsilon}),
        \end{dcases}    
    \end{equation*}
    where $p_{\varepsilon} = \exp(\varepsilon / 2)$,
    \begin{align*}
        [l_{x,\varepsilon},r_{x,\varepsilon}) &= 
        \begin{dcases}
            [0, 2C) &\text{if} \ x \in [0, C), \cr
            x + [-C, C) &\text{if } x \in [C, 1-C), \cr
            [1-2C, 1) &\text{otherwise}, \cr  
        \end{dcases}
    \end{align*}
    with $C=(\exp(\varepsilon / 2) - 1) / (2\exp(\varepsilon) - 2)$,
    is optimal for $[0,1)\to[0,1)$ under the absolute error and square error metric.
\end{theorem}
\begin{proof}
    Provided by analytical deduction
    with $[a,b) = [\tilde{a},\tilde{b}) = [0,1)$, $\loss=|y-x|$ and $\loss=|y-x|^2$ respectively.
    Appendix~\ref{appendix:optimal_concretization} provides the full proof.
\end{proof}

This optimality on $[0,1)$ can be transformed to $\mechanism':[a,b) \to [a,b)$ by applying 
$\trans: x'=(b-a)x+a,\, p'_{\varepsilon} = p_{\varepsilon} / (b-a),\, [l', r') = (b-a)\cdot[l, r) + a$,
while maintaining the optimality.

\begin{example}
    Figure~\ref{fig:optimal_classical} shows two examples of Theorem~\ref{theo:optimal_concretization}.
    When $x = 0$ in the left figure, the optimal GPM has $p \approx 1.64$ and $[l,r) \approx [0,0.38)$.
    When $x = 0.5$ in the right figure, the optimal GPM has $p \approx 1.64$ and $[l,r) \approx [0.31,0.69)$.
\end{example}

\begin{figure}[t]
    \centering
    \begin{subfigure}[b]{0.48\linewidth}
        \centering
        \includegraphics[width=0.98\linewidth]{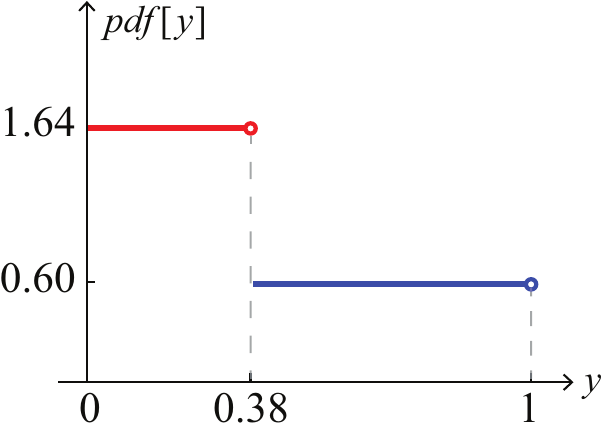}
        \caption{$x = 0$.}
    \end{subfigure}
    \hfill
    \begin{subfigure}[b]{0.48\linewidth}
        \centering
        \includegraphics[width=0.98\linewidth]{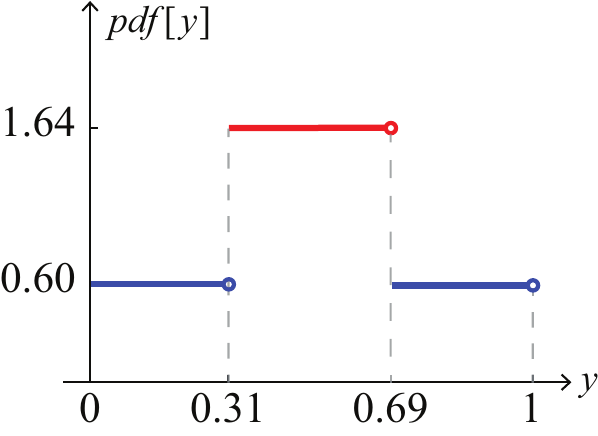}
        \caption{$x = 0.5$.}
    \end{subfigure}
    \caption{Optimal GPM (Theorem~\ref{theo:optimal_concretization}) when $\varepsilon = 1$, $x = 0$ and $x=0.5$.}
    \label{fig:optimal_classical}
\end{figure}

\label{anchor:classical_mse}
\textbf{MSE analysis.} 
We can calculate the mean squared error (MSE) of the optimal GPM in Theorem~\ref{theo:optimal_concretization} as 
\begin{equation*} 
    \begin{split}
        \mathrm{MSE}[\mechanism(x)] =& \int_{\outdomain} (y-x)^2 \cdot pdf[\mechanism(x) = y] \mathrm{d}y \\
        =& \int_{0}^{l_{x,\varepsilon}} (y-x)^2 \frac{p_{\varepsilon}}{\exp(\varepsilon)} \mathrm{d}y + \int_{l_{x,\varepsilon}}^{r_{x,\varepsilon}} (y-x)^2 p_{\varepsilon} \mathrm{d}y \\
        +& \int_{r_{x,\varepsilon}}^{1} (y-x)^2 \frac{p_{\varepsilon}}{\exp(\varepsilon)} \mathrm{d}y \;.
    \end{split}
\end{equation*}
which leads to the results in Appendix~\ref{appendix:classical_mse}.

The theoretical MSE allows us to analytically compare the optimal GPM with existing mechanisms.
For mechanisms defined on $\indomain = [0,1)$, e.g. SW~\cite{DBLP:conf/sigmod/Li0LLS20},
we can compare their MSE with the above result.
For mechanisms defined on other domains, e.g. PM~\cite{DBLP:conf/icde/WangXYZHSS019} on $\indomain = [-1,1)$,
we can transform the optimal GPM to $[-1,1)$ and compare their MSE.
Detailed comparisons with them are presented in the evaluation section.

\section{Optimal GPM for Circular Domain} \label{sec:cyclic_data}

This section presents the optimal GPM for the circular domain, another type of bounded domain.
Circular domains are widely used in cyclic data such as time, angle, and compass direction.
However, none of the existing piecewise-based mechanisms consider this type of domain,
limiting their applicability.

\textbf{Different meanings of distance.}
In the circular domain $[0,2\pi)$, the distance between two elements differs from that in the classical domain $[0,2\pi)$.
For example, if we denote the distance between $x$ and $y$ in the circular domain as $\lossmod(y, x) = |y-x|$,
then it implies $\lossmod(2\pi, 0) = 0$ and $\lossmod(3\pi/2, 0) = \pi/2$, which are different from $\loss(y, x) = |y-x|$ in the classical domain.
\ The biggest difference is that there are no endpoints in the circular domain.
This unique property makes the mechanisms designed for the classical domain not directly suitable for the circular domain.
Although we can ``flatten'' the circular domain to the classical domain, 
this conversion changes the distance between elements, leading to data utility loss.


Formally, in the circular domain $[0,2\pi)$, the distance metric $\lossmod(y, x)$ has the following relationship with $\loss(y, x)$ in the classical domain:
\begin{equation*}
    \lossmod(y, x) = \min \big(\loss(y,x), \loss(y, 2\pi - x)\big),
\end{equation*}
i.e. the distance between $y$ and $x$ is the smaller one between two arcs from $y$ to $x$.
Under this distance metric, finding the optimal $m$-GPM for the circular domain 
is to solve $\mechanism:[0,2\pi)\to [0,2\pi)$ such that
\begin{equation} \label{equ:circular_min-max}
    \begin{gathered}
        \min_{p_i, l_{i,x}, r_{i,x}} \max_{x} \int_{0}^{2\pi} \lossmod(y, x)\mathcal{P}_{\mechanism(x)}\mathrm{d}y, \\
        \text{s.t.} \ \mechanism \ \text{satisfies Definition~\ref{def:general_piecewise}},
    \end{gathered}
\end{equation}
and use the solved optimal $x$-independent $p_i$ to determine $l_{i,x}, r_{i,x}$ for any given $x$:
\begin{equation} \label{equ:circular_lr}
    \begin{gathered}
        \min_{l_{i,x}, r_{i,x}} \int_{0}^{2\pi} \lossmod(y, x)\mathcal{P}_{\mechanism(x)}\mathrm{d}y \\
        \text{s.t.} \ \mechanism \ \text{satisfies Definition~\ref{def:general_piecewise} with } p_i. 
    \end{gathered}
\end{equation}
We show that these two problems can be reduced to those in the classical domain, 
thereby enabling the usage of existing results.

\subsection{Reduced Forms}

Similar to the classical domain, the min-max objective of Formulation~(\ref{equ:circular_min-max}) also has a closed-form solution,
which reduces the problem to the classical domain.

\begin{lemma} \label{lemma:lossmod}
    The objective of Formulation~(\ref{equ:circular_min-max}) can be reduced to
    \begin{equation*}
        \min_{p_i, l_{i,x}, r_{i,x}} \int_{0}^{2\pi} \loss(y, \pi)\mathcal{P}_{\mechanism(\pi)}\mathrm{d}y. \\
    \end{equation*}
\end{lemma}

\begin{proof}
    (Sketch) We prove it by showing that, for any fixed $y$, 
    $\max_x\lossmod(y, x) = \lossmod(y, \pi) = \loss(y, \pi)$,
    i.e. the maximum distance between $y$ and $x$ is achieved at a unique $x = \pi$ for any $y$.
    Appendix~\ref{appendix:lossmod} provides the full proof.
\end{proof}
Following this reduction, the optimal results in the classical domain $[0,2\pi)$ can be applied to determine the optimal $p_i$.

For Formulation~(\ref{equ:circular_lr}) to solve $l_{i,x}$ and $r_{i,x}$, we shift the domain $[0,2\pi)$ by $\pi-x$.
This is a trick operation as $\lossmod(y, x)$ is transformed to the following form:
\begin{equation*}
    \begin{split}
        &\lossmod(y+\pi-x,x + \pi - x)  = \lossmod(y+\pi-x,\pi) \\
        =& \min \left(\loss(y+\pi-x,\pi), \loss(y+\pi-x, 2\pi - \pi)\right) \\
        =& \loss(y+\pi-x,\pi),
    \end{split}
\end{equation*}
which transforms $\lossmod$ to $\loss$ at $x = \pi$.
Then, the optimization problem in the shifted domain becomes
\begin{equation*} \label{equ:circular_reduced_prob}
    \min_{l_{i,x}, r_{i,x}} \int_{\pi-x}^{3\pi-x} \loss(y+\pi-x, \pi)\mathcal{P}_{\mechanism(\pi)}\mathrm{d}y.
\end{equation*}
It is a problem in the classical domain $[\pi-x,3\pi-x)$.
We can obtain the closed-form optimal $l_{i,x}$ and $r_{i,x}$ 
by applying the results of the classical domain.
Since the obtained $l_{i,x}$ and $r_{i,x}$ depends on $x$ in the shifted domain,
we shift them back to the circular domain using
\begin{align*}
    l_{i,x}^{\text{mod}} &= l_{i,x} - (\pi - x) \mod 2\pi, \\
    r_{i,x}^{\text{mod}} &= r_{i,x} - (\pi - x) \mod 2\pi.   
\end{align*}
Transformation invariants ensure their optimality in the circular domain.
Figure~\ref{fig:circular-reduced} summarizes the above two reductions
to solve the optimal $p_i$, $l_{i,x}^{\text{mod}}$ and $r_{i,x}^{\text{mod}}$ in the circular domain.

\begin{figure}[t]
    \centering
    \includegraphics[width=\linewidth]{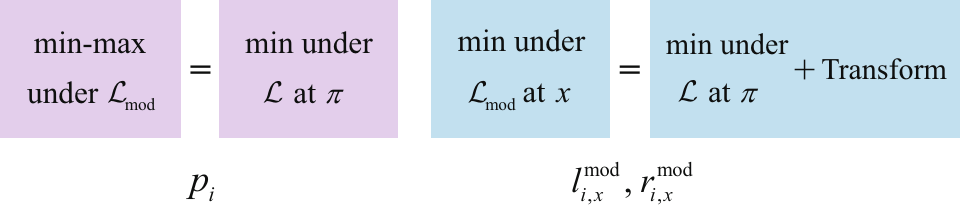}
    \caption{Reduced forms of solving the optimal $p_i$, $l_{i,x}^{\text{mod}}$ and $r_{i,x}^{\text{mod}}$.
    Optimizations under circular distance $\lossmod$ can be reduced to those under linear distance $\loss$.}
    \label{fig:circular-reduced}
\end{figure}

\subsection{Closed Form for Circular Domain}

By applying the above reductions and following the same steps as in the classical domain,
we can obtain the optimal GPM for the circular domain.

\begin{theorem} \label{theo:circular_concretization}
    If Hypothesis~\ref{theo:optimal_form} holds, then
    GPM $\mathcal{M}: [0,2\pi) \to [0,2\pi)$ with the following 
    closed-form instantiation
    \begin{equation*}
        pdf[\mathcal{M}(x) = y] = 
        \begin{dcases}
            p_{\varepsilon} & \text{if} \ y \in [l_{x,\varepsilon}^{\textup{mod}}, r_{x,\varepsilon}^{\textup{mod}}), \cr
            p_{\varepsilon} / \exp{(\varepsilon)} & y \in [0,2\pi) \, \backslash \, [l_{x,\varepsilon}^{\textup{mod}}, r_{x,\varepsilon}^{\textup{mod}}),
        \end{dcases}    
    \end{equation*}
    where $p_{\varepsilon} = \frac{1}{2\pi} \exp(\varepsilon/2)$,
    \begin{align*}
        l_{x,\varepsilon}^{\textup{mod}} &= \left( x - \pi\frac{\exp(\varepsilon/2) - 1}{\exp(\varepsilon) - 1} \right) \mod 2\pi, \\
        r_{x,\varepsilon}^{\textup{mod}} &= \left( x + \pi\frac{\exp(\varepsilon/2) - 1}{\exp(\varepsilon) - 1} \right) \mod 2\pi,
    \end{align*}
    is optimal for the circular domain under the absolute error and square error metric.
\end{theorem}
\begin{proof}
    Combination of the reduced forms, the results of the classical domain, Lemma~\ref{lemma:lossmod}, and transformation invariants
    leads to the conclusion.
\end{proof}

Compared to the optimal GPM for the classical domain in Theorem~\ref{theo:optimal_concretization},
the key difference is that the mechanism in the circular domain allows
$[l_{x,\varepsilon}^{\textup{mod}}, r_{x,\varepsilon}^{\textup{mod}})$ to span the $0$ or $2\pi$ boundary,
significantly reducing the error.
We also observe that their instantiations of $p_{\varepsilon}, l_{x,\varepsilon}, r_{x,\varepsilon}$
are connected through transformation invariants.
For instance, moving from the classical domain $[0,1)$ to the circular domain $[0,2\pi)$,
$p_{\varepsilon} = \exp(\varepsilon/2)$ transforms to
$p_{\varepsilon} = \frac{1}{2\pi} \exp(\varepsilon/2)$, reflecting the ratio of $[0,1)$ to $[0,2\pi)$.

\begin{example}
    Figure~\ref{fig:angular_mechanism} shows two examples of Theorem~\ref{theo:circular_concretization}
    when $\varepsilon = 1$.
    For $x = 0$ in the left figure, it samples the output $y$ from $[1.62\pi, 2\pi) \cup [0, 0.38\pi)$ with probability density $0.26$
    and from $[0.38\pi, 1.62\pi)$ with probability density $0.09$.
\end{example}

\begin{figure}[t]
    \centering
    \begin{subfigure}[b]{0.35\linewidth}
        \centering
        \includegraphics[width=0.98\linewidth]{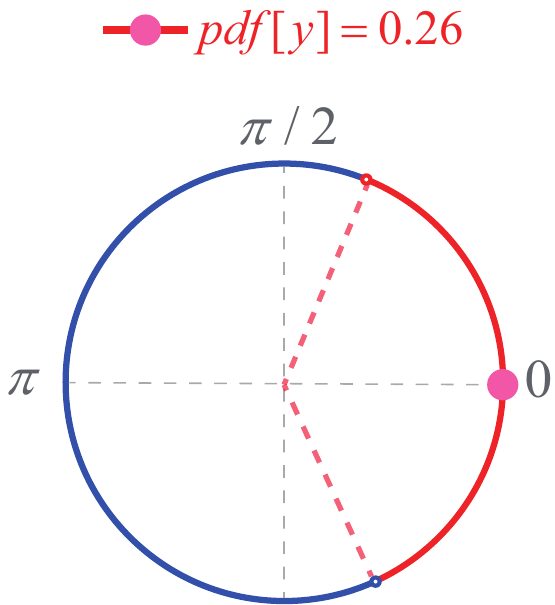}
        \caption{$x = 0$.}
    \end{subfigure}
    \hspace{0.1\linewidth}
    \begin{subfigure}[b]{0.35\linewidth}
        \centering
        \includegraphics[width=0.98\linewidth]{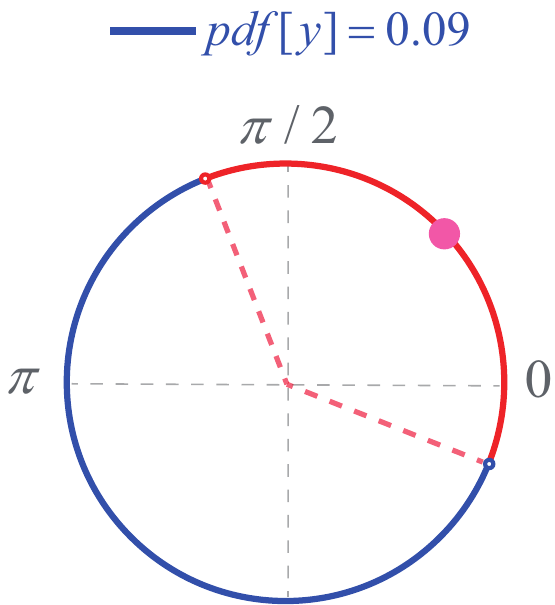}
        \caption{$x = \pi/4$.}
    \end{subfigure}
    \caption{Optimal GPM (Theorem~\ref{theo:circular_concretization}) for the circular domain with $\varepsilon = 1$.
    Here $[0,2\pi)$ is wrapped into a circle. Angle values in the red arc (centered at $x$) have a higher sampling pdf.}
    \label{fig:angular_mechanism}
\end{figure}

\textbf{MSE analysis.}
The MSE of the optimal GPM in Theorem~\ref{theo:circular_concretization} needs a separate analysis due to the circular domain.
The biggest difference from the classical domain is that there exist no endpoints, i.e. fixed farthest points,
in the circular domain. 
Without loss of generality, assume $\loss \coloneq |y-x|^2$ and $x > \pi$, the farthest distance from $x$ is always $\pi$,
i.e. from $x$ to $x - \pi$.
If we shift the data domain by $\pi - x$, then point $x - \pi$ is mapped to $0$.
This domain shift does not change the value of $\loss$, but $x$ now locates at $\pi$.
\ Therefore, the MSE of the optimal GPM in Theorem~\ref{theo:circular_concretization} has
an identical value for all $x$ in the circular domain, which is
\begin{equation*}
    \begin{split}
        & \mathrm{MSE}[\mechanism(x)] = \mathrm{MSE}[\mechanism(\pi)] \\
        = & 2\left(\int_{0}^{l_{\pi,\varepsilon}^{\textup{mod}}} (y-\pi)^2 \frac{p_\varepsilon}{\exp(\varepsilon)} \mathrm{d}y + \int_{l_{\pi,\varepsilon}^{\textup{mod}}}^{\pi} (y-\pi)^2 p_\varepsilon \mathrm{d}y\right).
    \end{split}
\end{equation*}
Calculating the above integral, we can obtain the MSE of the optimal GPM in the circular domain.

\begin{theorem} \label{theo:circular_mse}
    The optimal GPM in Theorem~\ref{theo:circular_concretization} has an identical MSE for all $x$ in the circular domain, which is
    \begin{equation*}
        \mathrm{MSE}[\mechanism(x)] = \frac{2}{3} \left((\pi^3 - C^3)\frac{p_\varepsilon}{\exp(\varepsilon)}  + C^3 p_\varepsilon \right),
    \end{equation*}
    where $C = \pi(\exp(\varepsilon/2) - 1)/ (\exp(\varepsilon) - 1)$.
\end{theorem}

\begin{proof}
    We have shown that the MSE at $x$ is the same as that at $\pi$.
    Then the calculation is straightforward by plugging in the values of $p_\varepsilon$ and $l_{\pi,\varepsilon}^{\textup{mod}}$ in Theorem~\ref{theo:circular_concretization}.
\end{proof}

If we do not consider $\lossmod$ and directly apply the optimal GPM in the classical domain (or SW and PM mechanisms)
to the circular domain, i.e. flatten the circular domain to the classical domain $[0,2\pi)$ and consider $\loss$,
the MSE will varies for different $x$.
In the flattened domain, the worst-case MSE is at $x = 0$ or $x = 2\pi$ and the best-case MSE is at $x = \pi$.
Therefore, the optimal GPM in Theorem~\ref{theo:circular_concretization} always has a lower MSE than ``flattened'' mechanisms.
The evaluation section will further demonstrate this.

\section{Distribution and Mean Estimation} \label{sec:estimation}

This section applies our mechanisms to support the commonly used distribution and mean estimation.
We will show that our mechanisms also provide optimality for these estimation tasks.

Assume a set of users with sensitive data $\mathcal{X} = \{x_1, x_2, \dots, x_n\}$.
They apply $\mechanism$ to produce randomized outputs $\mathcal{Y} = \{y_1, y_2, \dots, y_n\}$. 
The data collector then estimates the distribution and mean of $\mathcal{X}$ using $\mathcal{Y}$.
Specifically, the collector uses the values in $\mathcal{Y}$ as it is, i.e.
without knowing and applying any post-processing based on prior knowledge of $\mathcal{X}$.

\subsection{Distribution Estimation}
To estimate the distribution of $\mathcal{X}$ from $\mathcal{Y}$, the collector must discretize the continuous domain 
$\indomain$ into $k$ bins $B_1, B_2, \dots, B_k$.
Each bin $B_j$'s probability is then estimated by the proportion of $y_i$ that falls within $B_j$. 
Probabilities of all bins together form the estimated distribution $\hat{\mathcal{F}}_B$ using $\mathcal{Y}$.
This estimation's accuracy is measured by the distance between the estimated distribution $\hat{\mathcal{F}}_B$ and the 
true distribution $\mathcal{F}_B$ from $\mathcal{X}$~\cite{DBLP:conf/sigmod/Li0LLS20}.

Note that the bin size can impact the estimation accuracy due to the rounding of the bins.
However, it is actually a hyperparameter that does not inherently affect the estimation accuracy
if we have a sufficiently large number of bins, as the support of $\hat{\mathcal{F}}_B$ (i.e. non-zero bins)
converges to $\mathcal{Y}$
and the support of $\mathcal{F}_B$ converges to $\mathcal{X}$, i.e.
\begin{equation*}
    \lim_{k \to \infty} \mathrm{supp}(\hat{\mathcal{F}}_B) =  \mathcal{Y}, \quad \lim_{k \to \infty} \mathrm{supp}(\mathcal{F}_B) =  \mathcal{X}.
\end{equation*}
Our mechanisms guarantee the optimal error between $\mathcal{X}$ and $\mathcal{Y}$ for each $x_i$.
This ensures their optimality among GPM when applied to distribution estimation.

Some other statistical estimations are special applications of distribution estimation,
such as range and quantiles~\cite{DBLP:conf/sigmod/Li0LLS20} to estimate a specific part of the distribution.
Our mechanisms have optimality for these estimations as well.

\subsection{Mean Estimation}

To estimate the mean of $\mathcal{X}$ from $\mathcal{Y}$,
the collector uses the estimator $\hat{\mu} = \sum_{i=1}^{n}y_i/n$.
The accuracy of this estimator is measured by $|\hat{\mu}-\mu|$,
where $\mu$ is the true mean of $\mathcal{X}$.
Our mechanisms guarantee the optimal error between each $x_i$ and $y_i$,
which in turn leads to the smallest $|\hat{\mu}-\mu|$ among all GPMs under this metric.

Typically, a mean estimator may also need to be unbiased, i.e. $\mathrm{E}[\hat{\mu}] = \mu$.
This constraint translates to $\mathrm{E}[y_i] = \mathrm{E}[\mechanism(x_i)] = x_i$.\footnote{
    Actually, unbiasedness for numerical data is not as important as for categorical data,
    as it is for a single data point $x_i$.
    When the dataset $\mathcal{X}$ is not concentrated around a single point,
    an unbiased mechanism may not necessarily provide better performance.
}
This is unachievable for same-domain mapping $\mechanism:\indomain \to \indomain$ on classical domains,
as the endpoints of $\indomain$ cannot be the mean value (or center) of any distribution over $\indomain$.
For example, for any LDP mechanism $\mechanism: [0,1) \to [0,1)$, distribution of $\mechanism(0)$ can not be unbiased.
So the mechanism in Theorem~\ref{theo:optimal_concretization} is biased.

\textbf{Unbiased mean estimation.} 
Note that an unbiased mean estimator can be achieved by enlarging the output domain $\indomain\to \outdomain_\varepsilon$.
Mathematically, this involves incorporating the unbiasedness constraint $\mathrm{E}[\mechanism(x)] = x$ into optimization problems for solving $\mechanism$.
Following the same optimization process as in the classical domain, 
we hypothesize that the $3$-GPM remains optimal for domain $\outdomain_\varepsilon$.

\begin{hypothesis} \label{hypo:optimal_form}
    For any domain $\indomain \to \outdomain_\varepsilon$, 
    where $\outdomain_\varepsilon$ is a variable w.r.t $\varepsilon$,
    and under absolute error and square error metrics,
    the optimal piecewise-based mechanism falls into $3$-GPM.
\end{hypothesis}

Hypothesis~\ref{hypo:optimal_form} is a natural extension of Hypothesis~\ref{theo:optimal_concretization}, as the output domain $\outdomain_\varepsilon$ becomes explicit once $\varepsilon$ is specified.
Under the $3$-GPM, an unbiased mechanism $\mechanism$ with a variable output domain $\outdomain_\varepsilon$ can be analytically derived by incorporating the unbiasedness constraint.
As a complement to Theorem~\ref{theo:optimal_concretization},
we provide Theorem~\ref{theo:unbiased_optimal} for mean estimation in the classical domain.

\begin{theorem} \label{theo:unbiased_optimal}
    Denote $\outdomain_\varepsilon = [-C, C + 1)$ with 
    $C = (\exp(\varepsilon/2) + 1)/(\exp(\varepsilon/2) - 1)$.
    If Hypothesis~\ref{hypo:optimal_form} holds, then
    among the unbiased GPM $\mechanism: \indomain \to \outdomain_\varepsilon$ (i.e. $\mathrm{E}[\mechanism(x)] = x$),
    closed form
    \begin{equation*}
        pdf[\mathcal{M}(x) = y] =
        \begin{dcases}
            p_{\varepsilon}                       & \text{if} \ y \in [l_{x,\varepsilon}, r_{x,\varepsilon}), \cr
            p_{\varepsilon} / \exp{(\varepsilon)} & y \in \outdomain_\varepsilon \, \backslash \, [l_{x,\varepsilon}, r_{x,\varepsilon}),
        \end{dcases}
    \end{equation*}
    where $p = \exp(\varepsilon/2)/ (2C + 1)$,
    \begin{align*}
        l_{x,\varepsilon} &= \frac{C+1}{2} \cdot x - \frac{(3C+1)(C-1)}{4C}, \\
        r_{x,\varepsilon} &= \frac{C+1}{2} \cdot x + \frac{(C+1)(C-1)}{4C}.
    \end{align*}
    is optimal for $[0,1)\to \outdomain_\varepsilon$ and the square error metric.
\end{theorem}

\begin{proof}
    Instantiations of $C$, $p$, $l_{x,\varepsilon}$, and $r_{x,\varepsilon}$ are derived by analytical deduction.
    Appendix~\ref{appendix:unbiased_optimal} proves the unbiasedness.
\end{proof}

In the context of the circular domain, the $\mechanism$ in Theorem~\ref{theo:circular_concretization} is unbiased, as 
the distribution of $\mechanism$
is always centered at $x$. 
This property is also illustrated in Figure~\ref{fig:angular_mechanism}.

Table~\ref{tab:estimation} summarizes the optimal GPM for both 
distribution and mean estimation in three domain types.
Theorems use a concrete domain $\indomain$ for the sake of clarity.
We do not give closed-form optimal GPM for distribution estimation on $\indomain\to \outdomain_\varepsilon$ 
because $\outdomain_\varepsilon$ can not be easily concretized by constraints as in mean estimation.

\begin{table}[t]
    \begin{center}
        \caption{Optimal distribution and mean estimation in three domain types under GPM.}
        \label{tab:estimation}
        \begin{tabular}{l l l}
            \toprule
                                                          & \textbf{Distribution}                            & \textbf{Mean}                              \\
            \midrule
            Classical ($\indomain \to \indomain$)              & Theorem~\ref{theo:optimal_concretization}                   & Theorem~\ref{theo:optimal_concretization} (biased)    \\
            \vspace{-0.6em} \\  
            Circular domain                   & Theorem~\ref{theo:circular_concretization}                  & Theorem~\ref{theo:circular_concretization} (unbiased) \\
            \vspace{-0.6em} \\
            Classical ($\indomain \to \outdomain_\varepsilon$) & \makecell[c]{--} & Theorem~\ref{theo:unbiased_optimal} (unbiased) \\
            \bottomrule
        \end{tabular}
    \end{center}
\end{table}

\section{Discussion and Extension} 

\subsubsection*{Minimize Error at a Specific $x$} \label{sec:optimize_x}

Our proposed optimal GPMs are designed to minimize the worst-case error over the whole domain.
However, the framework can be used to minimize the error at any specific $x$.
This is useful when the data distribution is concentrated around a specific data point.

Formally, assume the data distribution is concentrated around $x_0$, and we want to minimize the error at $x_0$.
The optimization problem in Lemma~\ref{lemma:endpoint} for solving $p_{i}$ can be modified to
\begin{equation*}
    \min_{p_i, l_{i,x}, r_{i,x}} \int_{\outdomain} \loss(y, x_0)\mathcal{P}_{\mechanism(x_0)}\mathrm{d}y.
\end{equation*}
This optimization problem generally leads to a different $p_i$ from the optimal GPM for the worst-case error.

\begin{example}
    Assume $\indomain \to \outdomain = [0,1) \to [0,1)$, $\loss \coloneq |y-x|$, and the data distribution is concentrated around $x_0 = 0.2$.
    Solving the above optimization problem with $\varepsilon = 1$ gives probability (of the second piece) $p_2 = 1.54$ and $Err(x_0) = 0.241$.
    In contrast, the optimal GPM for the worst-case error uses $p_2 = 1.64$, as shown in Figure~\ref{fig:optimal_classical},
    which gives $Err(x_0) = 0.243$.
\end{example}

Importantly, optimizing for a specific $x_0$ does not leak information about $x_0$,
as the mechanism (i.e. $p_i$, $l_{i,x}$, and $r_{i,x}$) still contains no information about $x_0$.
Moreover, observing $Err(x)$ at all $x$ does not reveal $x_0$,
as the error at $x_0$ is not necessarily the smallest.
Therefore, the adversary cannot infer $x_0$ from observing the mechanism.

\subsubsection*{An Extension: 2D Polar Coordinates} \label{sec:polar}

Polar coordinates are widely used in relative location representation,
e.g. navigation systems that have the locations of surrounding objects relative to it.
Our proposed optimal GPM for the classical and circular domain can be combined and 
extended to polar coordinates for collecting such data under LDP.

\textbf{Privacy.} A polar coordinate data is represented by a 2D tuple $(x_1, x_2) \in [0, d) \times [0, 2\pi)$,
where $x_1$ is the distance from the pole and $x_2$ is the angle from the polar axis.
The first dimension is linear, while the second is naturally circular,
thus we can combine the optimal GPM for both domains to provide LDP for such data.

\textbf{Utility.} Our mechanisms guarantee the optimal error for each dimension.
Therefore, if we use $\loss_{\text{2D}} \coloneq \loss(y_1, x_1) + \lossmod(y_2, x_2)$ as the error metric for the 
polar coordinate data and optimally assign the privacy parameter $\varepsilon$ to each dimension, 
the optimal GPM preserves the optimal error.

\textbf{Optimal assignment of $\varepsilon$.}
Formally, we want to assign $\varepsilon = \varepsilon_1 + \varepsilon_2$ to $x_1$ and $x_2$ respectively to 
minimize the worst-case error in 2D polar coordinates.
This error is the sum of the worst-case error for $x_1$ and $x_2$ under $\loss$ and $\lossmod$ respectively.
Therefore, the optimal assignment of $\varepsilon$ can be derived by solving
\begin{equation*}
    \min_{\varepsilon_1, \varepsilon_2} Err_{\text{1,wor}}(\varepsilon_1) + Err_{\text{2,wor}}(\varepsilon_2),
\end{equation*}
where $Err_{\text{1,wor}}(\varepsilon_1)$ and $Err_{\text{2,wor}}(\varepsilon_2)$ are the worst-case error
for the classical domain $[0, d)$ and the circular domain $[0, 2\pi)$ respectively.
Closed-form $Err_{\text{1,wor}}(\varepsilon_1)$ and $Err_{\text{2,wor}}(\varepsilon_2)$ can be derived from 
instantiations of the mechanism and the error metric.
Then the above optimization problem gives the optimal $\varepsilon_1$ and $\varepsilon_2$.
Appendix~\ref{appendix:polar_assignment_epsilon} provides details for solving this optimization problem.

\begin{example}
    Figure~\ref{fig:polar_mechanism} shows two examples of the optimal GPM for 2D polar coordinates in $[0, 1) \times [0, 2\pi)$ with $\varepsilon = 1 + 2\pi$.
    The green point represents the sensitive data, the optimal GPM samples the output from the pink area with a higher probability.
\end{example}

\begin{figure}[t]
    \centering
    \begin{subfigure}[b]{0.35\linewidth}
        \centering
        \includegraphics[width=0.98\linewidth]{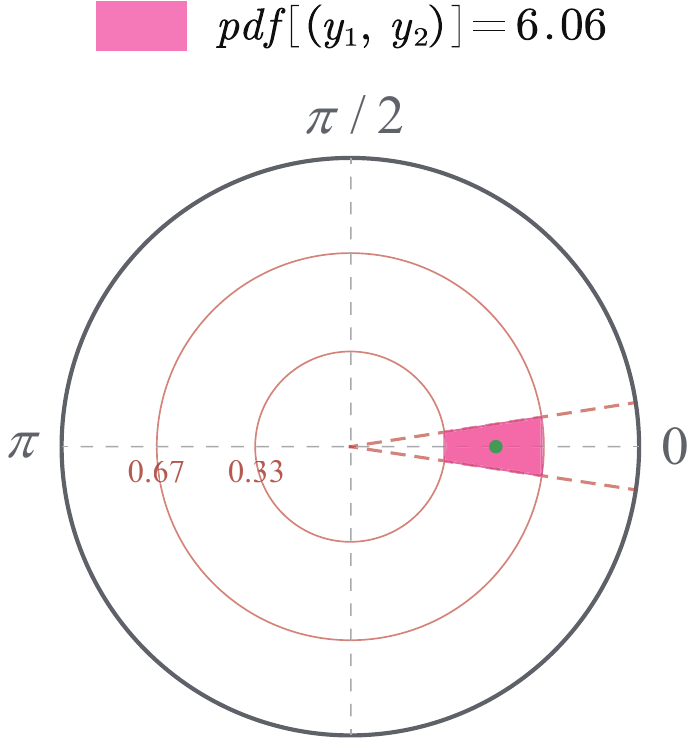}
        \caption{$(x_1, x_2) = (0.5, 2\pi)$.}
    \end{subfigure}
    \hspace{0.1\linewidth}
    \begin{subfigure}[b]{0.35\linewidth}
        \centering
        \includegraphics[width=0.98\linewidth]{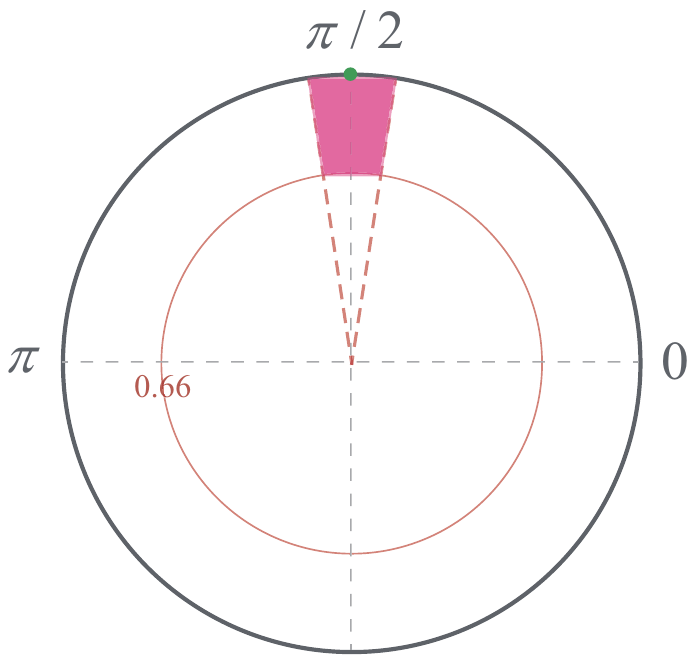}
        \caption{$(x_1, x_2) = (1, \pi/2)$.}
    \end{subfigure}
    \caption{Optimal GPM for 2D polar coordinates when $\varepsilon = 1 + 2\pi$ and $\loss = |y-x|^2$
    under $\loss_{\text{2D}} \coloneq \loss(y_1, x_1) + \lossmod(y_2, x_2)$.}
    \label{fig:polar_mechanism}
\end{figure}

\section{Evaluations} \label{sec:evaluations}

\begin{figure*}[t]
    \centering
    \begin{minipage}{0.47\textwidth}
        \centering
        \begin{subfigure}[b]{0.48\linewidth}
            \includegraphics[width=0.98\linewidth]{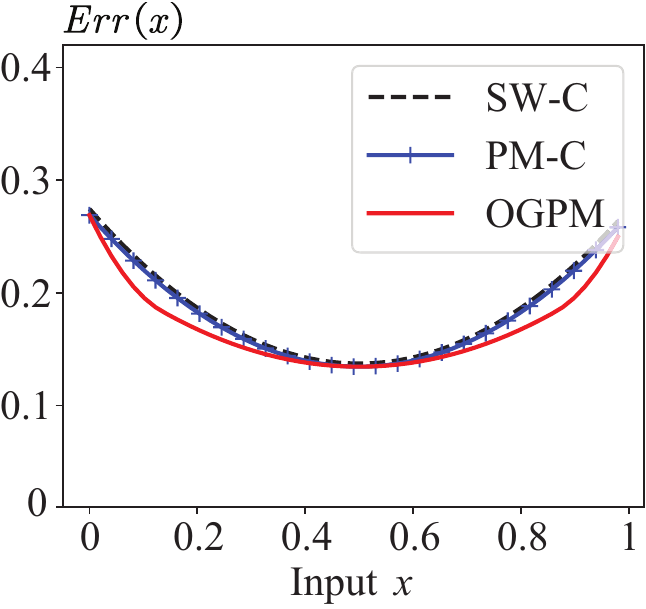}
            \caption{Privacy parameter $\varepsilon = 2$.}
        \end{subfigure}
        \hfill
        \begin{subfigure}[b]{0.48\linewidth}
            \includegraphics[width=0.98\linewidth]{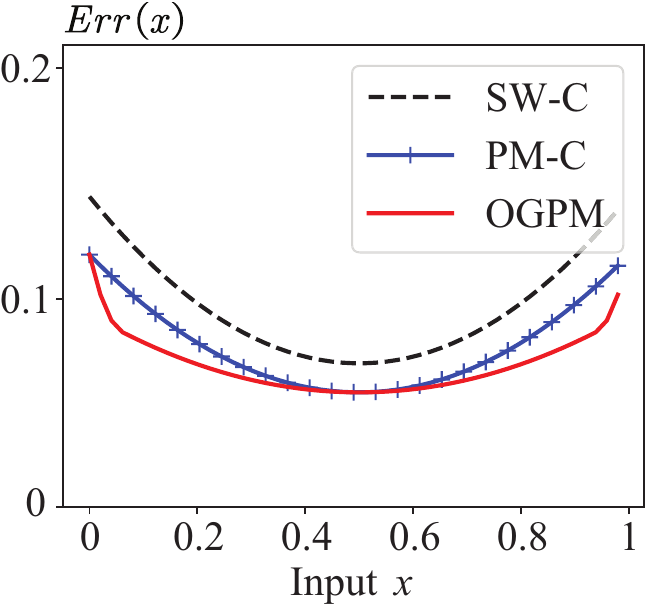}
            \caption{Privacy parameter $\varepsilon = 4$.}
        \end{subfigure}
        \caption{Whole-domain error comparison in the classical domain with error metric $\loss = |y-x|$.}
        \label{fig:exp:classical_whole_domain}
    \end{minipage}
    \hfill
    \begin{minipage}{0.47\textwidth}
        \centering
        \begin{subfigure}[b]{0.48\linewidth}
            \includegraphics[width=0.98\linewidth]{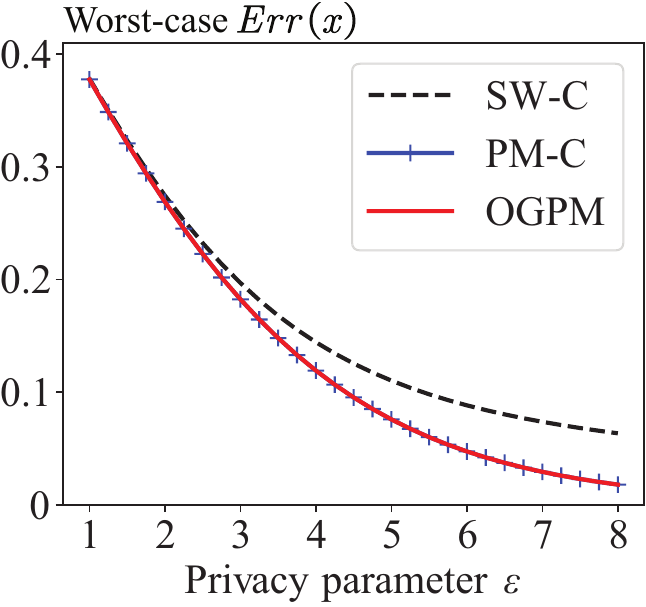}
            \caption{Distance metric $\loss = |y-x|$.}
        \end{subfigure}
        \hfill
        \begin{subfigure}[b]{0.48\linewidth}
            \includegraphics[width=0.98\linewidth]{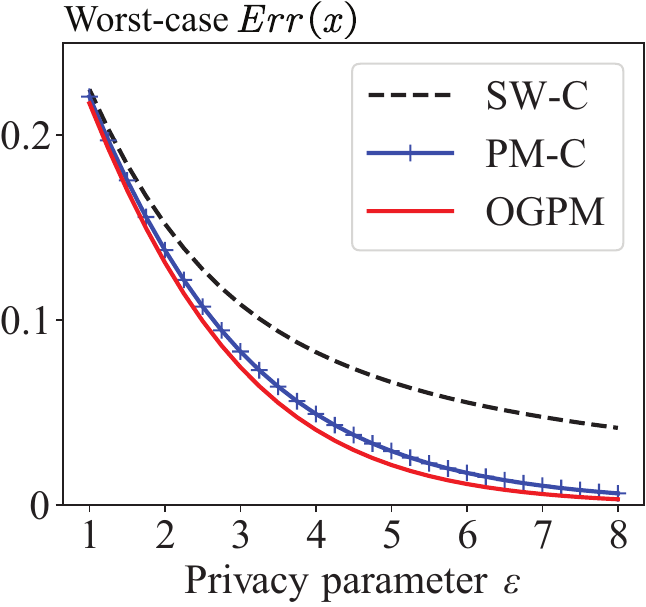}
            \caption{Distance metric $\loss = |y-x|^2$.}
        \end{subfigure}
        \caption{Worst-case error comparison in the classical domain.}
        \label{fig:exp:classical_worst_case}
    \end{minipage}
\end{figure*}

\begin{figure*}[t]
    \centering
    \begin{minipage}{0.47\textwidth}
        \centering
        \begin{subfigure}[b]{0.48\linewidth}
            \includegraphics[width=0.98\linewidth]{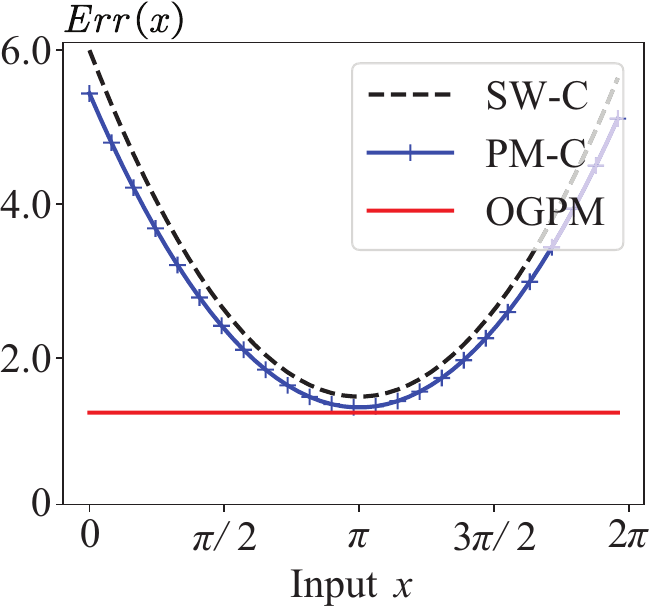}
            \caption{Privacy parameter $\varepsilon = 2$.}
        \end{subfigure}
        \hfill
        \begin{subfigure}[b]{0.48\linewidth}
            \includegraphics[width=0.98\linewidth]{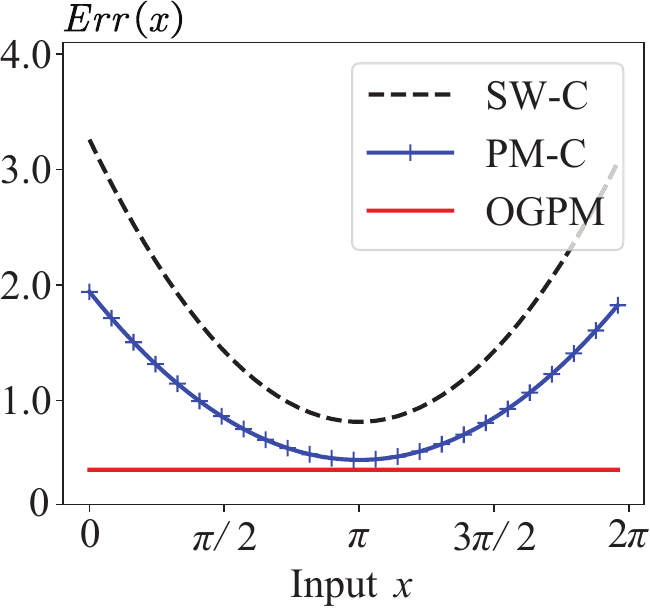}
            \caption{Privacy parameter $\varepsilon = 4$.}
        \end{subfigure}
        \caption{Whole-domain error comparison in the circular domain with error metric $\loss = |y-x|^2$.}
        \label{fig:exp:circular_whole_domain}
    \end{minipage}
    \hfill
    \begin{minipage}{0.47\textwidth}
        \centering
        \begin{subfigure}[b]{0.48\linewidth}
            \includegraphics[width=0.98\linewidth]{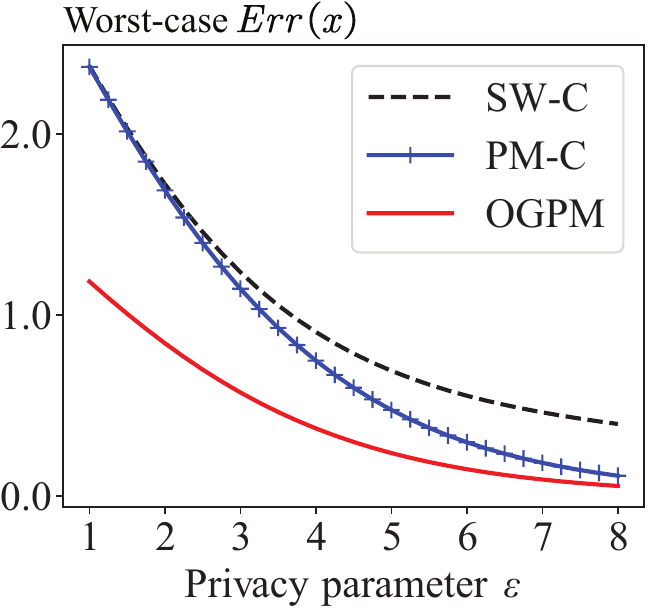}
            \caption{Distance metric $\loss = |y-x|$.}
        \end{subfigure}
        \hfill
        \begin{subfigure}[b]{0.48\linewidth}
            \includegraphics[width=0.98\linewidth]{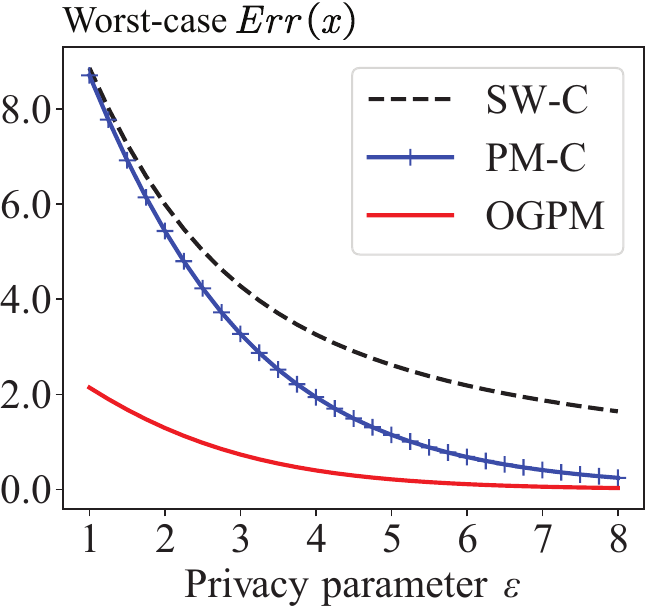}
            \caption{Distance metric $\loss = |y-x|^2$.}
        \end{subfigure}
        \caption{Worst-case error comparison in the circular domain.}
        \label{fig:exp:circular_worst_case}
    \end{minipage}
\end{figure*}

This section evaluates the theoretical and experimental data utility of our methods by comparing them with existing 
instantiations and their variants:
\begin{itemize}
    \item OGPM: closed-form optimal GPM (Theorem~\ref{theo:optimal_concretization} and~\ref{theo:circular_concretization}
    for the classical and circular domain, respectively).
    \item OGPM-U: unbiased closed-form optimal GPM (Theorem~\ref{theo:unbiased_optimal}) for mean estimation in the classical domain.
    \item PM~\cite{DBLP:conf/icde/WangXYZHSS019}, SW~\cite{DBLP:conf/sigmod/Li0LLS20}, and their post-processed versions:
    PM is the first TPM designed for mean estimation, while SW is designed for distribution estimation. 
    Both mechanisms output enlarged domains but can be post-processed by truncating outputs to the input domain. 
    These post-processed versions are referred to as T-PM and T-SW for convenience.
    \item PM-C and SW-C: the compressed versions of PM and SW for $\indomain \to \indomain$.
    For the best potential of PM and SW, we adapt them to $\indomain \to \indomain$ as PM-C and SW-C by linearly compressing their output domain $\outdomain_\varepsilon$ to $\indomain$,
    i.e. by transformation invariants, which maintains the privacy level.
\end{itemize}
We also compare OGPM's expected error with non-piecewise-based mechanisms that can be applied to bounded numerical domains:
\begin{itemize}
    \item Variants of the Laplace mechanism: including the staircase mechanism~\cite{DBLP:conf/isit/GengV14}, Laplace mechanism with post-processing by truncation (T-Laplace),
    and the bounded Laplace mechanism (B-Laplace)~\cite{DBLP:journals/jpc/HolohanABA20}, which redesigns a bounded Laplace-shape distribution.
    \item Purkayastha mechanism~\cite{DBLP:conf/ccs/WeggenmannK21}: a mechanism for directional data on spheres $\mathbb{S}^{n-1}$. When $n=2$, it is a counterpart of OGPM in the circular domain.
\end{itemize}

We omit comparison with PTT~\cite{ma10430379} because it does not provide a concrete method to find 
a closed-form mechanism.
We use $\indomain=[0,1)$ as the classical domain; this does not change
their data utility ordering by the transformation invariant.
\ PM and SW can only be applied to the classical domain, 
so when evaluating them in the circular domain,
we ``flatten'' the circular domain to the classical domain $[0,2\pi)$ and apply them to the flattened domain.

The first subsection presents the comparison of expected errors, followed by the distribution and mean estimations on real-world datasets in the second subsection.

\subsection{Expected Errors} \label{subsec:theoretical_evaluation}

GPM's data utility under distance metric $\loss$ 
is measured by the expected error $Err(x)$ in Formula~(\ref{equ:utility_1}):
\begin{equation*}
    Err(x) = \int_{\outdomain} \mathcal{L}(y, x)\mathcal{P}_{\mathcal{M}(x)} \mathrm{d}y,
\end{equation*}
where $x$ is the sensitive input and $y$ is the output of $\mechanism$.
Based on $Err(x)$, two types of error need to be considered:
\begin{itemize}
    \item \textbf{whole-domain error:} $Err(x)$ values for the whole domain, i.e. for all $x\in \indomain$.
    \item \textbf{worst-case error:} the largest $Err(x)$ value among the whole domain. 
    PM~\cite{DBLP:conf/icde/WangXYZHSS019} and PTT~\cite{ma10430379} also use this error to evaluate data utility.
\end{itemize}
We have proved that the worst-case error of the classical domain is from the endpoints, 
so this error is actually the error at $x=0$ and $x=1$ for the classical domain.
For the circular domain, the worst-case error is at $x = \pi$.

\subsubsection{Comparison with PM-C and SW-C} \label{subsec:comparison_with_PM_SW}
PM-C and SW-C exhibit the best potential of PM and SW, so we compare OGPM with them in the first place.
The comparisons are conducted under the classical domain and the circular domain, respectively.

\textbf{Classical Domain.}
Figure~\ref{fig:exp:classical_whole_domain} shows the comparison of the whole-domain error in the classical domain.
We use distance metric $\loss = |y-x|$ and set $\varepsilon=2$ and $\varepsilon=4$ for the comparison.
It can be seen that OGPM consistently has the lowest error across all $x$ values.
\ For all mechanisms, the expected error achieves the maximal value at the endpoints
and the minimal value at the midpoint.
At small $\varepsilon$ values, their errors are similar due to the strong randomness (privacy constraint).
At larger $\varepsilon$ values, OGPM's error has a significant advantage over PM-C and SW-C,
especially between the endpoints and the midpoint. 
\ Statistically, under $\loss = |y-x|$ with $\varepsilon = 2$, OGPM's average error is $94.2\%$ of PM-C and $92.3\%$ of SW-C across the whole domain.
When $\varepsilon = 4$, OGPM's average error is $90.5\%$ of PM-C and $74.7\%$ of SW-C.
More comparisons under smaller $\varepsilon$ values are provided in Appendix~\ref{appendix:smaller_epsilon}.

Figure~\ref{fig:exp:classical_worst_case} shows the comparison of the worst-case error w.r.t $\varepsilon$ in the classical domain.
The error is measured by two distance metrics: $\loss = |y-x|$ and $\loss = |y-x|^2$.
It can be seen that the error of all mechanisms decreases with $\varepsilon$, but OGPM and PM-C decreases faster than SW-C.
\ For $\loss = |y-x|$, OGPM has almost the same worst-case error as PM-C; for $\loss = |y-x|^2$, OGPM's error is slightly smaller than PM-C.
For both metrics, OGPM's worst-case error is the lowest across all $\varepsilon$ values.
\ Statistically, 
under $\loss = |y-x|^2$, OGPM's average error is $89.9\%$ of PM-C and $61.7\%$ of SW-C.

\textbf{Circular Domain.}
Figure~\ref{fig:exp:circular_whole_domain} shows the comparison of the whole-domain error in the circular domain.
We use distance metric $\loss = |y-x|^2$ and set $\varepsilon=2$ and $\varepsilon=4$ for the comparison.
It can be seen that OGPM consistently has the lowest error across all $x$ values, and the error is stable across $x$,
which is consistent with the theoretical analysis in Theorem~\ref{theo:circular_mse}.
\ For PM-C and SW-C, which treat the circular domain as the classical domain, their errors vary with $x$ and are higher than OGPM's error, especially near the endpoints.
\ Statistically, under $\loss = |y-x|^2$ with $\varepsilon = 2$, OGPM's average error is $47.5\%$ of PM-C and $43.0\%$ of SW-C across the whole domain.
When $\varepsilon = 4$, OGPM's average error is $41.3\%$ of PM-C and $24.6\%$ of SW-C.

Figure~\ref{fig:exp:circular_worst_case} shows the comparison of the worst-case error w.r.t $\varepsilon$ in the circular domain.
The error is measured by two distance metrics: $\loss = |y-x|$ and $\loss = |y-x|^2$.
Similar to the classical domain, OGPM has the lowest worst-case error across all $\varepsilon$ values,
and the advantage is more significant, especially for small $\varepsilon$ values.
\ Statistically, under $\loss = |y-x|$, OGPM's average error is $50.0\%$ of PM-C and $41.6\%$ of SW-C across the range of $\varepsilon$.
Under $\loss = |y-x|^2$, OGPM's average error is $22.4\%$ of PM-C and $15.4\%$ of SW-C.

The above comparisons show that OGPM has the lowest expected error in both the classical and circular domains.
The advantage of OGPM is more significant in the circular domain, where the error is stable across $x$.

\subsubsection{Comparison with PM and SW} \label{subsec:comparison_with_PM_SW_original}
Figure~\ref{fig:exp:ablation} presents the comparison of the whole-domain error in the classical domain for the original PM and SW mechanisms, along with their post-processed versions, T-PM and T-SW. 
For a fair comparison, OGPM is adapted to the domain $\indomain = [-1,1)$ to match PM's design, 
while SW and OGPM remain consistent with $\indomain = [0,1)$. 
The post-processing of PM and SW involves truncating their outputs in the enlarged domain to the input domain, i.e. applying $\mathcal{I}\circ \mechanism(x)$, where $\mathcal{I}: \outdomain \to \indomain$ is the truncation operator.
We use the distance metric $\loss = |y-x|^2$ and set $\varepsilon=2$ for the comparison among these five mechanisms. 
It can be observed that OGPM consistently achieves the lowest error across all $x$ values, with a more significant advantage compared to the comparison with PM-C and SW-C. 
This is because the original PM and SW output larger domains, resulting in higher errors. 
Meanwhile, T-PM reduces the error of PM more effectively than T-SW reduces the error of SW, 
as the original PM has a more enlarged output domain than SW, making truncation more impactful. 
This comparison highlights OGPM's error advantage over the original PM, SW, and their post-processed versions when applied to their respective data domains.

\begin{figure}
    \centering
    \begin{subfigure}[b]{0.465\linewidth}
        \includegraphics[width=0.98\linewidth]{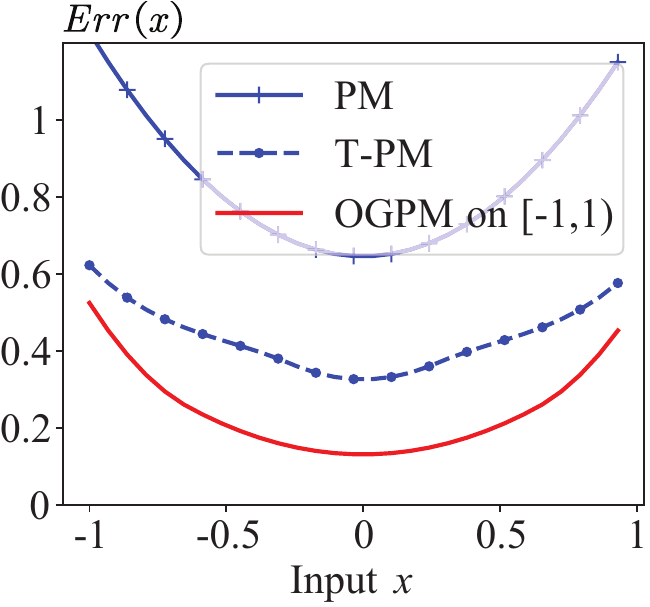}
        \caption{Comparison with PM.}
    \end{subfigure}
    \hfill
    \begin{subfigure}[b]{0.465\linewidth}
        \includegraphics[width=0.98\linewidth]{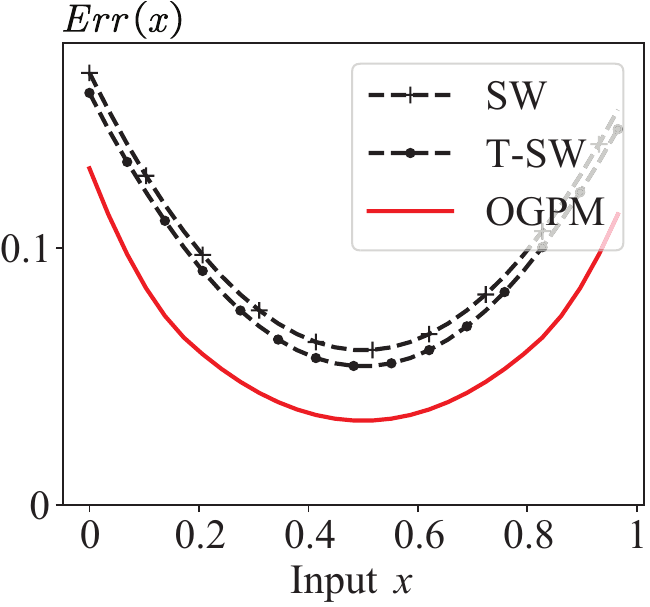}
        \caption{Comparison with SW.}
    \end{subfigure}
    \caption{Whole-domain error comparison with PM and SW on their data domains (i.e. $\indomain = [-1,1)$ and $\indomain = [0,1)$, respectively) when $\varepsilon = 2$.}
    \label{fig:exp:ablation}
\end{figure}

\begin{figure}
    \centering
    \begin{subfigure}[b]{0.465\linewidth}
        \includegraphics[width=0.98\linewidth]{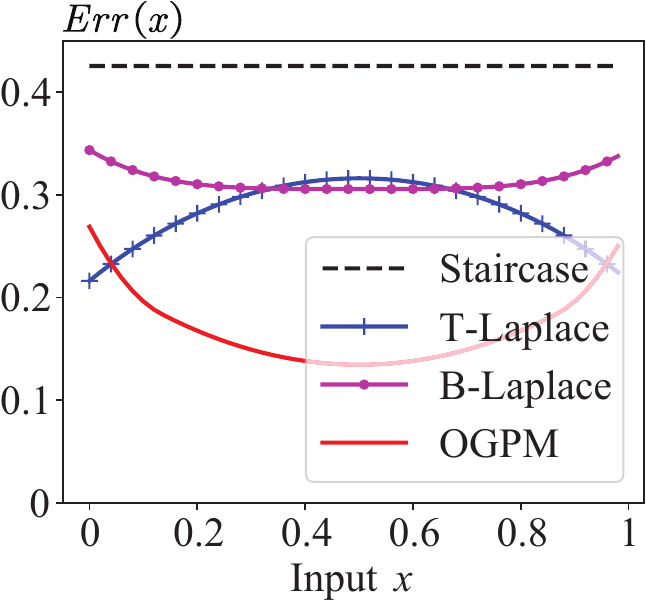}
        \caption{Privacy parameter $\varepsilon = 2$.}
    \end{subfigure}
    \hfill
    \begin{subfigure}[b]{0.465\linewidth}
        \includegraphics[width=0.98\linewidth]{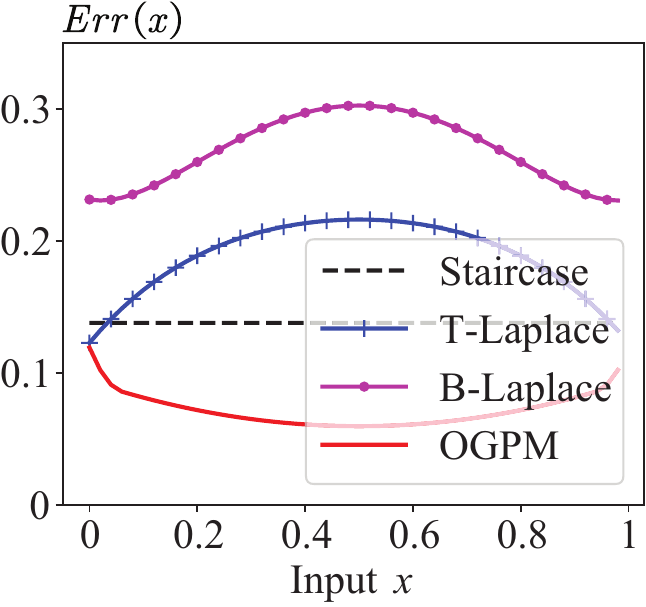}
        \caption{Privacy parameter $\varepsilon = 4$.}
    \end{subfigure}
    \caption{Whole-domain error comparison with the staircase mechanism~\cite{DBLP:conf/isit/GengV14},
    T-Laplace and B-Laplace mechanisms~\cite{DBLP:journals/jpc/HolohanABA20} in the classical domain with error metric $\loss = |y-x|$.}    
    \label{fig:exp:staircase_truncated}
\end{figure}

\subsubsection{Comparison with the Staircase Mechanism, T-Laplace, and B-Laplace} \label{subsec:staircase_truncated}
In addition to piecewise-based mechanisms, the Laplace mechanism and its variants can also be applied to the classical domain to achieve LDP. 
Among these, the staircase mechanism~\cite{DBLP:conf/isit/GengV14} claims to be optimal under certain assumptions. 
For the input domain $\indomain = [0,1)$ (i.e. sensitivity $\Delta = 1$) and error metric $\loss = |y-x|$, its expected error is given by Theorem 3 in~\cite{DBLP:conf/isit/GengV14}:
$\exp(\varepsilon / 2) / (\exp(\varepsilon) - 1)$.
\ Another approach involves using the Laplace mechanism with truncation~\cite{DBLP:journals/jpc/HolohanABA20}, referred to here as T-Laplace for convenience. 
T-Laplace preserves the privacy guarantees of the Laplace mechanism while reducing the expected error, particularly for data points near the endpoints or for small $\varepsilon$ values.
Additionally, the bounded Laplace mechanism (B-Laplace)~\cite{DBLP:journals/jpc/HolohanABA20} introduces a redesigned bounded Laplace-shaped distribution tailored for bounded domains.\footnote{
    Appendix~\ref{appendix:b_laplace} provides details on the expected error of B-Laplace.
}

\begin{figure}
    \begin{minipage}{0.465\linewidth}
        \centering
        \includegraphics[width=0.98\linewidth]{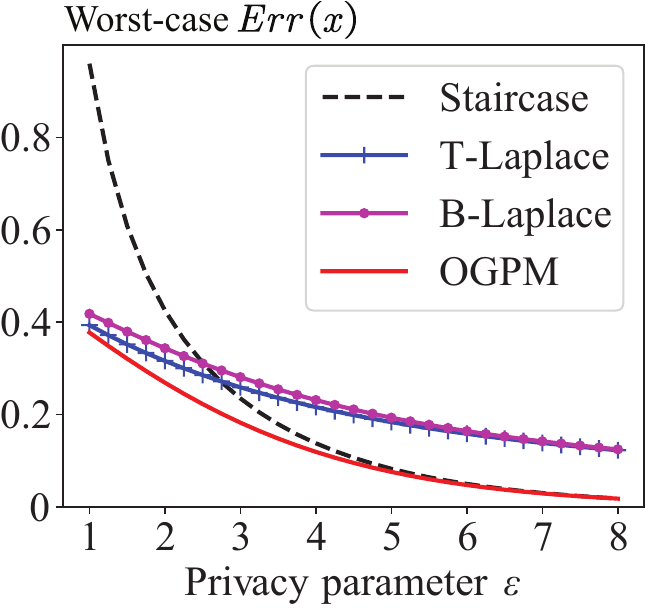}
        \caption{Worst-case error comparison (continued fr\-om Figure~\ref{fig:exp:staircase_truncated}).}   
        \label{fig:exp:staircase_truncated_worst_case}
    \end{minipage}
    \hfill
    \begin{minipage}{0.465\linewidth}
        \centering
        \includegraphics[width=0.98\textwidth]{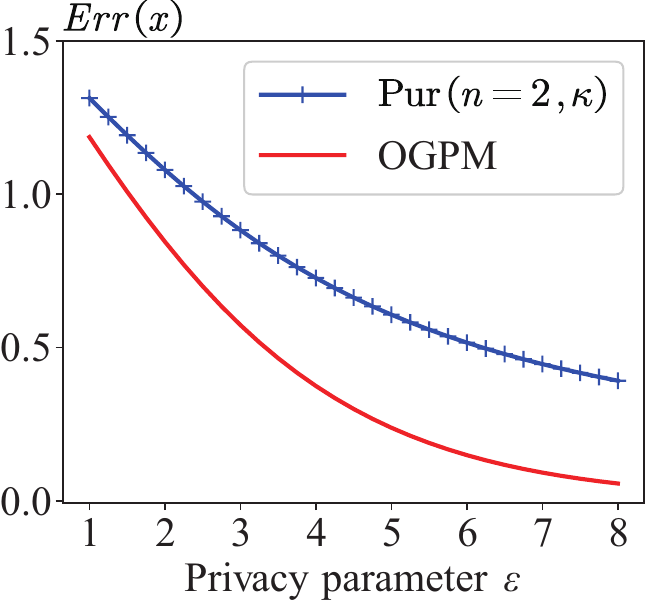}
        \caption{Comparison with the Purkayastha mechani\-sm~\cite{DBLP:conf/ccs/WeggenmannK21} for sphere $\mathbb{S}^{n-1}$.}
        \label{fig:exp:comparison_purkayastha}
    \end{minipage}
\end{figure}

Figure~\ref{fig:exp:staircase_truncated} compares the whole-domain error in the classical domain $\indomain = [0,1)$ 
for the staircase mechanism, T-Laplace, and B-Laplace. 
These mechanisms exhibit distinct error patterns across the domain.
\ For the staircase mechanism, the error remains constant, as it is determined by a fixed staircase distribution and is independent of $x$.
For T-Laplace, the error reaches its maximum at the midpoint and its minimum at the endpoints, 
as truncation favors the endpoints. 
For instance, when $x=0$, it is error-free with a probability of $1/2$, 
due to the symmetry of the Laplace distribution around $0$.
For B-Laplace, the error trend varies with $\varepsilon$. 
When $\varepsilon = 2$, the error decreases with $x$ and reaches its minimum at the midpoint, 
whereas for $\varepsilon = 4$, the error increases with $x$ and peaks at the midpoint.
Despite these differing error patterns, OGPM generally achieves lower errors than the staircase mechanism, T-Laplace,
and B-Laplace across the whole domain.

Figure~\ref{fig:exp:staircase_truncated_worst_case} compares the worst-case error w.r.t. $\varepsilon$ in the classical domain.
OGPM consistently achieves the lowest worst-case error across all $\varepsilon$ values.
For small $\varepsilon$, T-Laplace and B-Laplace exhibit a significant advantage over the staircase mechanism; however, this advantage diminishes as $\varepsilon$ increases.
At larger $\varepsilon$ values, the error of the staircase mechanism approaches that of OGPM.

\subsubsection{Comparison with the Purkayastha Mechanism}
The paper ``Differential Privacy for Directional Data''~\cite{DBLP:conf/ccs/WeggenmannK21} introduces two mechanisms for data on spheres $\mathbb{S}^{n-1}$: 
the VMF mechanism (ensuring indistinguishability of any two points with distance \emph{through} the sphere) 
and the Purkayastha mechanism (ensuring indistinguishability of any two points with distance \emph{along} the sphere). 
When $n = 2$, the sphere $\mathbb{S}^{1}$ corresponds to a circle, making the Purkayastha mechanism a counterpart of OGPM in the circular domain. 
Therefore, we compare them in the circular domain.\footnote{
    We omit the comparison with the VMF mechanism also because (i) it has been shown that the Purkayastha mechanism outperforms the VMF mechanism (with the same sensitivity $\Delta_{\measuredangle} = \pi$ for sphere $\mathbb{S}^{1}$,
    e.g. Figure 5 and 10 in~\cite{DBLP:conf/ccs/WeggenmannK21}), and (ii) the expected error of the VMF mechanism lacks a closed-form expression (Theorem 17 in~\cite{DBLP:conf/ccs/WeggenmannK21}), 
    making it complex to compute.}

Figure~\ref{fig:exp:comparison_purkayastha} presents the comparison of the expected error in the circular domain between OGPM and the Purkayastha mechanism. 
The expected error of the Purkayastha mechanism is derived using the closed-form expressions in Theorem 19 and 22 of~\cite{DBLP:conf/ccs/WeggenmannK21}, 
with $\kappa = \varepsilon / \Delta_{\measuredangle}$. 
Since the errors of both mechanisms are $x$-independent in the circular domain, it suffices to compare their worst-case errors. 
The results demonstrate that OGPM consistently outperforms the Purkayastha mechanism, achieving significantly lower errors.

\subsection{Distribution and Mean Estimations} \label{subsec:experimental_evaluation}

This section compares the experimental data utility of the mechanisms in distribution and mean estimations.

\subsubsection{Setup}
We choose the MotionSense dataset~\cite{noauthor_motionsense_nodate, Malekzadeh:2019:MSD:3302505.3310068} for the evaluation.
It contains smartphone sensor data recorded during various human activities.
Specifically, we use the data from the first three files, encompassing a total of $6\,159$ data entries.
We focus on two types of sensors:
\begin{itemize}
    \item Accelerometer (linear data): We normalize the dataset to $[0,1)$ for the classical domain.
    \item Attitude sensor (angular data):  We use this dataset for the circular domain.
\end{itemize}
Upon applying each LDP mechanism to these datasets, we evaluate the accuracy of distribution and mean estimations on them.
For distribution estimation, we divide the domain into $k = 50$ bins and
compute the distance between the estimated distribution ($\hat{\mathcal{F}}_B$) and the true distribution ($\mathcal{F}_B$) 
by summing the absolute difference of each bin's value. Formally,
\begin{equation*}
    \mid \hat{\mathcal{F}}_B - \mathcal{F}_B \mid \coloneq \sum_{i=1}^{50} \mid \hat{\mathcal{F}}_{B_i} - \mathcal{F}_{B_i} \mid,
\end{equation*}
where $\hat{\mathcal{F}}_{B_i}$ and $\mathcal{F}_{B_i}$ are the $i$-th bin's value of the estimated and true distributions, respectively.
\ Although this approach cannot capture the property of the circular data, it remains the most viable metric for comparing circular distributions~\cite{fisher_statistical_2000,mardia2009directional}. 
Under a more relevant approach, the performance of OGPM for the circular domain could be even better. 

For mean estimation, we compute the absolute difference between the estimated and true mean, i.e. $|\hat{\mu} - \mu|$,
where $\hat{\mu}$ is the estimated mean and $\mu$ is the true mean in the classical domain or the circular domain.
In the classical domain, the true mean is $\mu = \frac{1}{n}\sum_{i=1}^n x_i$.
In the circular domain, the mean is computed by the circular mean formula~\cite{fisher_statistical_2000,mardia2009directional}:
\begin{equation*}
    \mu = \text{atan2}\left(\frac{1}{n}\sum_{i=1}^n \sin x_i, \frac{1}{n}\sum_{i=1}^n \cos x_i\right).
\end{equation*}
We repeat the experiments $500$ times for stable results and report the average error.

\subsubsection{Distribution Estimation}

Figure~\ref{fig:exp:distribution_estimation} shows the comparison of the errors of distribution estimation in the classical and circular domains.
We can see that OGPM outperforms SW and PM with smaller errors in both types of domains.
\ In the classical domain, OGPM's error decreases faster when $\varepsilon$ increases above $3$,
and SW-C decreases slower than PM-C, consistent with the expected error comparison.
\ In the circular domain, OGPM's error is significantly lower than SW-C and PM-C,
despite the limitation of the distance metric used for circular distributions.
We also observe that SW-C performs better than PM-C in this domain.
This is because SW has higher sampling probabilities for both the central piece and other pieces, 
making it sample the true value more frequently when $\varepsilon$ is large in practice in a large-size domain,
despite the large expected error theoretically.
\ Statistically, OGPM's distribution error is $93.5\%$ of PM-C and $86.7\%$ of SW-C in the classical domain,
and $72.2\%$ of PM-C and $84.0\%$ of SW-C in the circular domain.

\subsubsection{Mean Estimation}
Figure~\ref{fig:exp:mean_estimation} shows the comparison of the errors of mean estimation in the classical and circular domains.
OGPM consistently outperforms other mechanisms in both types of domains, with significantly lower errors.
\ In the classical domain, we also compare with OGPM-U, which is specifically designed for unbiased mean estimation. 
Since the Accelerometer dataset is concentrated around zero, 
it particularly favors unbiased mechanisms, as their outputs tend to average closely to the true mean.
We can see that OGPM-U achieves significantly lower errors than OGPM when $\varepsilon$ is small.
\ In the circular domain, OGPM is unbiased, having a negligible mean estimation error.
We also observe that SW-C outperforms PM-C in this large-scale domain.
\ Statistically, OGPM's mean estimation error is $66.2\%$ of PM-C and $55.4\%$ of SW-C in the classical domain.
In the circular domain, OGPM's mean estimation error is merely $2.3\%$ of PM-C and $3.6\%$ of SW-C.


\begin{figure}[t]
    \centering
    \begin{subfigure}[b]{0.48\linewidth}
        \includegraphics[width=0.98\linewidth]{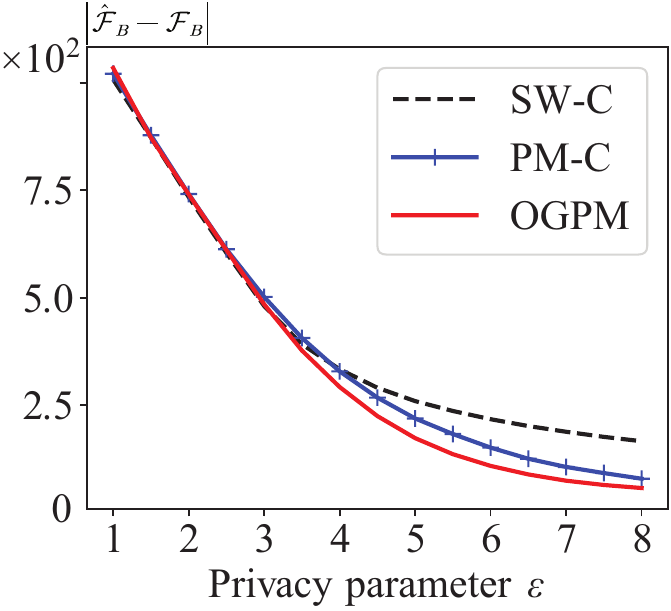}
        \caption{Classical domain.}
    \end{subfigure}
    \begin{subfigure}[b]{0.48\linewidth}
        \includegraphics[width=0.98\linewidth]{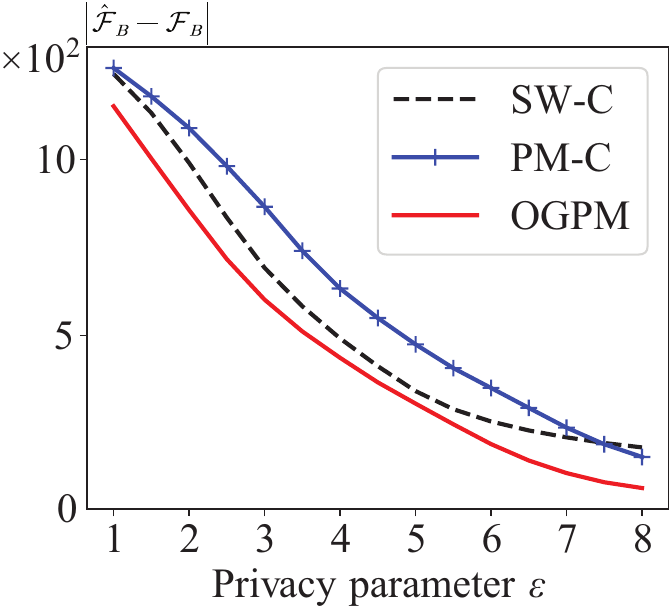}
        \caption{Circular domain.}
    \end{subfigure}
    \caption{Comparison of distribution estimation error.}
    \label{fig:exp:distribution_estimation}
\end{figure}

\begin{figure}[t]
    \centering
    \begin{subfigure}[b]{0.48\linewidth}
        \includegraphics[width=0.98\linewidth]{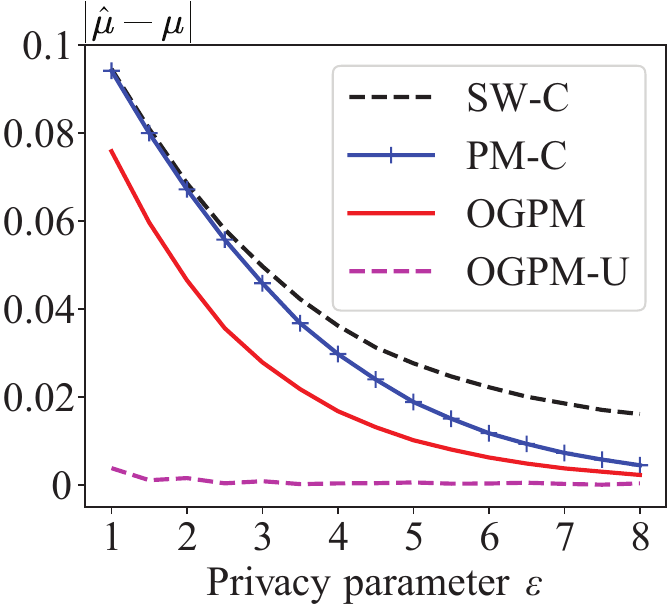}
        \caption{Classical domain.}
    \end{subfigure}
    \begin{subfigure}[b]{0.48\linewidth}
        \includegraphics[width=0.98\linewidth]{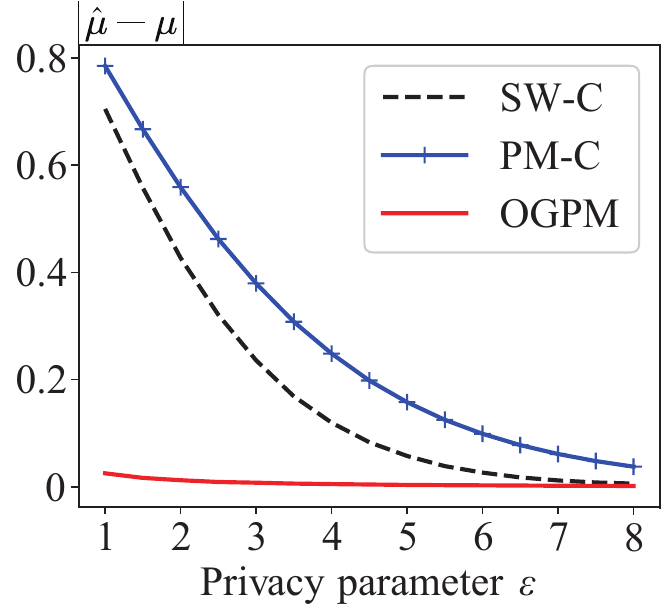}
        \caption{Circular domain.}
    \end{subfigure}
    \caption{Comparison of mean estimation error.}
    \label{fig:exp:mean_estimation}
\end{figure}

\section{Related Work}

This paper focuses on the optimal mechanism for collecting numerical data with bounded domains under LDP,
related to numerical data collection and the optimality of DP mechanisms.

\subsubsection*{Numerical Data Collection under LDP}

Classical noise-adding mechanisms such as the Laplace and Gauss mechanisms~\cite{DBLP:conf/tcc/DworkMNS06,DBLP:journals/fttcs/DworkR14}
add randomly sampled noise from a distribution to the data to achieve LDP.
However, they generate outputs in an unbounded domain due to the unbounded noise distributions,
rendering them unsuitable for applications requiring bounded domains~\cite{DBLP:conf/icde/WangXYZHSS019}.

For bounded data domains, the basic idea is to sample the output from a carefully designed distribution on the bounded domain.
Duchi et al. randomize any data in $[-1,1]$ to two discrete values $y \in \{-C_\varepsilon, C_\varepsilon\}$~\cite{DBLP:journals/corr/DuchiWJ16},
where $C_\varepsilon$ is a constant depending on the privacy level $\varepsilon$.
This binary-output mechanism exhibits a large randomization error as the output space is too coarse.
\ PM~\cite{DBLP:conf/icde/WangXYZHSS019} extends the binary output to a continuous output in $[-C_\varepsilon, C_\varepsilon]$.
It designs a piecewise-based mechanism to sample the output, which uses different sampling intervals for different $x$, achieving a lower randomization error.
Both PM and later SW~\cite{DBLP:conf/sigmod/Li0LLS20} have shown the data utility advantage of TPM in numerical data collection with bounded domains under LDP.
PTT~\cite{ma10430379} discusses the optimality of TPM.
It shows that there exist TPM instantiations that yield optimal data utility, but it does not provide the closed-form mechanism.
\ TPM is a special case of GPM using $3$-piece distributions and a specific form for each piece. 
Meanwhile, existing TPM instantiations are designed for specific error metrics and only classical domains.

\textbf{Applications of TPM.}
A natural application of TPM is in high-dimensional numerical data.
This includes scenarios where the sensitive data may be a high-dimensional vector~\cite{DBLP:conf/icde/WangXYZHSS019,DBLP:journals/compsec/ZhangNFHZ23},
an infinite data stream~\cite{DBLP:conf/sigmod/RenSYYZX22}, or matrixes~\cite{DBLP:journals/eswa/WangGRMZ23}.
\ Another application area is federated learning.
TPM ensures LDP in the $[0,1]$ domain,
which is commonly used as the normalized domain in model training.
TPM avoids the data clipping required by noise-adding mechanisms~\cite{DBLP:conf/ccs/AbadiCGMMT016,DBLP:journals/eswa/LiuTTDC23,DBLP:journals/corr/abs-2405-08299}.
\ Another typical application is in sensors, where the data is bounded by the sensor's physical nature.
This paper can replace the existing TPM to achieve better utility.

A recent work aims to design other types of bounded distributions besides TPM to achieve (L)DP~\cite{DBLP:conf/sp/ZhangZSTHYXC24}.
They tailor sin functions and quartic functions on the bounded domain to satisfy the DP constraint.
From their experimental results, the piecewise-based design is the best choice among their six instantiations for bounded domains under DP.
This also indicates the advantages of piecewise-based mechanisms in the bounded domain.

\subsubsection*{Optimality of DP Mechanisms}

Achieving optimal data utility is a common concern across all DP mechanisms.
The staircase mechanism showed that the Laplace mechanism does not generate optimal noise~\cite{DBLP:journals/isci/Soria-ComasD13,DBLP:conf/isit/GengV14}.
It samples noise from a staircase distribution, which has been shown to achieve lower error compared to the Laplace mechanism.
\ Besides optimality in concrete error values (which is the focus of this paper), another widely focused 
concept is asymptotic optimality~\cite{DBLP:journals/corr/abs-1902-04495,ma10430379},
which studies optimal asymptotic bounds of a mechanism or a statistical estimation.
Appendix~\ref{appendix:optimality_discussion} discusses more detailed optimalities.

Beyond the optimality of a single DP procedure, 
the optimality of multifold compositions such as iterative training under DP, is studied by advanced compositions~\cite{DBLP:journals/toc/MurtaghV18,DBLP:journals/tit/KairouzOV17,DBLP:journals/corr/abs-1905-02383,DBLP:conf/icml/DongD020}.
In high-dimensional settings, the sensitivity set across dimensions also influences the optimality~\cite{DBLP:conf/ccs/Xiao0D23}, 
because it affects the choice of the privacy parameter.
These works either focus on different privacy constraints or discuss optimality beyond a single DP procedure, making them orthogonal to ours. 



\section{Conclusions}

This paper presents the optimal piecewise-based mechanism for collecting numerical data with bounded domains under LDP.
To find the optimal mechanism among all possible piecewise mechanisms,
we generalize the existing $3$-piece mechanism to an $m$-piece mechanism with the most general form.
We proposed a framework that combines analytical proofs and off-the-shelf optimization solvers to find the optimal mechanism.
Our results include the closed-form optimal piecewise mechanisms for both the classical and circular domain.
Theoretical and experimental evaluations confirm the advantages of our mechanisms over existing mechanisms.


    
\clearpage

\begin{acks}
    We thank the anonymous reviewers and the revision editor for their valuable feedback and guidance, which significantly improved this paper.
    We also acknowledge the use of GPT-4 for language refinement in this paper.
    \ This research is supported in part by the U.S. National Science Foundation under grants CNS-2245689 (CRII) and DGE-2336252,
    as well as by the 2022 Meta Research Award for Privacy-Enhancing Technologies.
\end{acks}

\bibliographystyle{ACM-Reference-Format}
\bibliography{reference}
\appendix

\section{Proofs} \label{appendix:proofs}

\textbf{Notations.} Table~\ref{tab:notations} provides the notations used throughout this paper.

\subsection{Proof of Lemma~\ref{lemma:endpoint}} \label{appendix:endpoint}

\begin{proof}
    We prove the inner $\max_x$ problem has a closed form. According to the definition of $\mathcal{P}_{\mechanism(x)}$,
    it is
    \begin{equation*}
        \max_x\int_{\outdomain} \loss(y, x)\mathcal{P}_{\mechanism(x)}\mathrm{d}y = \max_x \sum_{i = 1}^{m} p_i\int_{l_i}^{r_i} \loss(y, x)\mathrm{d}y.
    \end{equation*}
    Denote $f_i(x)\coloneq \int_{l_i}^{r_i} \loss(y, x)\mathrm{d}y$, where $\loss(y, x) = |y-x|^p$.
    First, we prove that $f_i(x)$ is a convex function w.r.t. $x$.
    \ Specifically, based on the relationship between $x$ and $[l_i, r_i)$, the value of $x$ is split into three cases:
    (i) $x\in [a, l_i)$, (ii) $x\in [l_i, r_i)$, and (iii) $x\in [r_i, b)$.
    We prove the second derivative of $f_i(x)$ w.r.t. $x$ is non-negative in each case,
    thus $f_i(x)$ is convex.

    \textbf{Case (i):} $x\in [a, l_i)$.
    The integral is:
    \begin{equation*}
        \begin{split}
            f_i(x)&=\int_{l_i}^{r_i} |y-x|^p \mathrm{d}y=\int_{l_i}^{r_i} (y-x)^p \mathrm{d}y\\
            &=\frac{(r_i-x)^{p+1}-(l_i-x)^{p+1}}{p+1}.
        \end{split}
    \end{equation*}
    The second derivative w.r.t. $x$ is
    \begin{equation*}
            \frac{\partial^2}{\partial x^2} f_i(x) = p(r_i-x)^{p-1} - p(l_i-x)^{p-1} \geq 0.
    \end{equation*}
    The inequality holds because $(r_i-x)^{p-1} \geq (l_i-x)^{p-1}$ for $x\in [a, l_i)$.

    \textbf{Case (ii):} $x\in [l_i, r_i)$. The integral is
    \begin{equation*}
        \begin{split}
            f_i(x)&=\int_{l_i}^{r_i} |y-x|^p \mathrm{d}y\\
            &=\int_{l_i}^{x} (x-y)^p \mathrm{d}y+\int_{x}^{r_i} (y-x)^p \mathrm{d}y\\
            &=\frac{(x-l_i)^{p+1}}{p+1}+\frac{(r_i-x)^{p+1}}{p+1}.
        \end{split}
    \end{equation*}
    The second derivative w.r.t. $x$ is
    \begin{equation*}
        \begin{split}
            \frac{\partial^2}{\partial x^2} f_i(x) &= p(x-l_i)^{p-1} + p(r_i-x)^{p-1} \geq 0.
        \end{split}
    \end{equation*}
    The inequality holds because both $(x-l_i)^{p-1}$ and $(r_i-x)^{p-1}$ are non-negative for $x\in [l_i, r_i)$.

    \textbf{Case (iii):} $x\in [r_i, b)$.
    The integral is
    \begin{equation*}
        \begin{split}
            f_i(x)&=\int_{l_i}^{r_i} |y-x|^p \mathrm{d}y=\int_{l_i}^{r_i} (x-y)^p \mathrm{d}y\\
            &=\frac{(x-l_i)^{p+1}-(x-r_i)^{p+1}}{p+1}.
        \end{split}
    \end{equation*}
    The second derivative w.r.t. $x$ is
    \begin{equation*}
        \begin{split}
            \frac{\partial^2}{\partial x^2} f_i(x) &= p(x-l_i)^{p-1} - p(x-r_i)^{p-1} \geq 0.
        \end{split}
    \end{equation*}
    The inequality holds because $(x-l_i)^{p-1} \geq (x-r_i)^{p-1}$ for $x\in [r_i, b)$.

    The above three cases show that the second derivative of $f_i(x)$ w.r.t. $x\in [a,b)$ is non-negative.
    Thus, the non-negative weighted sum $\sum_{i = 1}^{m} p_i f_i(x)$ is also convex w.r.t. $x$ \cite{boyd2004convex}.
    \ According to the Bauer maximum principle~\cite{noauthor_bauer_2024}:
    any function that is convex attains its maximum at some extreme points of set.
    This means that the optimal $x$ is achieved at the endpoints of $x \in \indomain$, i.e. $x = a$ or $x = b$.
    Therefore, we have
    \begin{equation*}
        \max_x\int_{\outdomain} \loss(y, x)\mathcal{P}_{\mechanism(x)}\mathrm{d}y = \max_{\{a,b\}}\int_{\outdomain} \loss(y, x)\mathcal{P}_{\mechanism(x)}\mathrm{d}y,
    \end{equation*}
    which completes the proof.
\end{proof}

\emph{Remark:} 
This lemma can be empirically validated by the whole-domain error plots in Figure~\ref{fig:exp:classical_whole_domain} and Figure~\ref{fig:exp:circular_whole_domain},
where the maximum of the whole-domain error is achieved at the endpoints.

\begin{table}[t]
    \begin{center}
        \caption{Notations}\label{tab:notations}
        \begin{tabular}{l l}
            \toprule
            \textbf{Symbol}\hspace{3.5em} & \textbf{Description} \\
            \midrule
            $x$ & Sensitive input (from raw data) \\
            $y$ & Randomized output \\

            $\mathcal{D}$ & Input domain \\ 
            $\tilde{\mathcal{D}}$ & Output domain \\

            $pdf[\cdot]$ & Probability density function \\
            $\mathcal{P}_{\mechanism(x)}$ & Probability density function of $\mechanism(x)$ \\
            \midrule
            $p_{\varepsilon}$ & Sampling probability w.r.t. $\varepsilon$ \\
            $[l_{x, \varepsilon}, r_{x, \varepsilon})$ & Sampling interval w.r.t. $x$ and $\varepsilon$ \\
            \bottomrule
        \end{tabular}
    \end{center}
\end{table}

\subsection{Proof of Lemma~\ref{lemma:plus_piece}} \label{appendix:plus_piece}
\begin{proof}
    The optimal $(m+1)$-piecewise mechanism and the optimal $m$-piecewise mechanism may superficially differ due to the extra piece.
    Thus, we define a piece-merging operation to merge the redundant pieces.
    We will show that if the optimal $(m+1)$-piecewise mechanism is the same as the optimal $m$-piecewise mechanism after merging redundant pieces,
    then increasing $m$ does not decrease the optimal error, i.e. the optimal piece number is $m$.

    Assume the optimal $m$-piecewise mechanism is determined by the tuple set
    \begin{equation*}
        S_m = \{ (p_i, l_i, r_i) : i \in [m] \}.
    \end{equation*}
    To merge redundant pieces, we define a piece-merging operation:
    \begin{equation*}
        (p_i, l_i, r_i) \uplus (p_j, l_j, r_j) 
        \coloneq 
        \begin{dcases}
            (p_i, l_i, r_j) \quad \text{if} \ p_i = p_j \ \text{and} \ i+1 = j, \\ 
            \{(p_i, l_i, r_i), (p_j, l_j, r_j)\} \quad \text{otherwise}.
            \end{dcases}
    \end{equation*}
    Because the optimal ($m+1$)-piecewise mechanism is the same as the $m$-piecewise mechanism, it follows that
    \begin{equation*} 
        \uplus_{i,j = 1}^{m+1}S_{m+1} = \uplus_{i,j = 1}^{m}S_{m},
    \end{equation*} 
    where $\uplus_{i,j = 1}^{m}S_{m}$ merges all consecutive pieces with the same $p$.
    Denote the merged optimal $m$-piecewise mechanism as $\uplus_{i,j = 1}^{m}S_{m} \coloneq S_m^{\uplus}$ and the piece number as $|S_m^{\uplus}| = m^*$.
    Because both sides of the above equation are optimal, this means that 
    if $(p_k, l_k, r_k)$ is an arbitrary piece in the optimal ($m+1$)-piecewise mechanism, then merging it with $S_m^{\uplus}$ remains $S_m^{\uplus}$, i.e.
    \begin{equation*}
        (p_k, l_k, r_k) \uplus_{i = 1}^{m^*} S_m^{\uplus} = S_m^{\uplus}.
    \end{equation*}
    This premise indicates that there does not exist a piece $(p_k, l_k, r_k)$ besides $S_m^{\uplus}$ lowers the error.

    Without loss of generality, we can consider the optimal $(m+2)$-piecewise mechanism,
    which allows an extra piece besides the optimal $(m+1)$-piecewise mechanism.
    We claim that the extra piece is still captured by $S_m^{\uplus}$.
    \ The key insight is: adding an extra optimal piece to the optimal $(m+1)$-piecewise mechanism
    is the \emph{same} as adding it to the optimal $m$-piecewise mechanism,
    because the optimal $(m+1)$-piecewise mechanism is the same as the optimal $m$-piecewise mechanism.

    Since adding an extra optimal piece to the optimal $m$-piecewise mechanism remains $S_m^{\uplus}$,
    then for any $k \in [m+2]$, merging piece $(p_k, l_k, r_k)$ in the optimal ($m+2$)-piecewise mechanism
    remains an $m$-piecewise mechanism $S_m^{\uplus}$.
    
    For $m + 3$ or more, it follows the same logic. Adding an arbitrary piece is equivalent to adding it to the optimal $(m+2)$-piecewise mechanism,
    which is the same as adding it to the optimal $m$-piecewise mechanism.
    Thus, the optimal $m$-piecewise mechanism is the same as the optimal $(m+3)$-piecewise mechanism after merging redundant pieces,
    and so on.
\end{proof}

\emph{Remark:} Intuitively, the extra pieces (of the optimal $(m+1)$-piecewise mechanism and beyond) 
is similar to the redundant variables in optimization theory:
adding more non-negative variables to a minimization objective does not decrease the \emph{optimal} value.
Here the error from each piece is a variable, and it is non-negative, which leads to the same conclusion:
adding more pieces (one or more) to the optimal $m$-piecewise mechanism does not decrease the optimal error.

This lemma means that we can determine the optimal $m$-piecewise mechanism for $m = 1, 2, \ldots$, 
until the optimal $(m+1)$-piecewise mechanism is identical to the optimal $m$-piecewise mechanism for all $x$ and $\varepsilon$. 
This statement can be empirically validated by attempting to find counterexamples using larger $m$ than the optimal $m$. 
The source code of our framework provides scripts and results to empirically validate this lemma.

\subsection{Proof of Theorem~\ref{theo:transformation_invariant}} \label{appendix:transformation_invariant}

\begin{proof}
\textbf{Privacy invariant:} For any input $v, v' \in \indomain'$ and any output $y \in \outdomain'$:
\begin{equation*}
    \frac{pdf\left[\mechanism'(v) = y\right]}{pdf\left[\mechanism'(v') = y\right]} \leq \frac{p}{c} \div \frac{p}{\exp(\varepsilon)c}= \exp(\varepsilon).
\end{equation*}

\textbf{Utility invariant:} For any $x' = cx + d \in \indomain'$, we can calculate the error difference between $\mechanism'(x')$ and $\mechanism'_{\text{bad}}(x')$ as follows:
\begin{equation*}
    \begin{split}
        &Err(x', \mechanism') - Err(x', \mechanism'_{\text{bad}}) \\
        =&Err(cx+d, \mechanism') - Err(cx+d, \mechanism'_{\text{bad}}) \\
        =& \int_{\outdomain'} \loss(y,cx+d)\left(\mathcal{P}_{\mechanism'(cx+d)} - \mathcal{P}_{\mechanism_{\text{bad}}'(cx+d)}\right)\mathrm{d}y. \\
    \end{split}
\end{equation*}
Let $y_t = (y-d)/c$, then $\mathrm{d}y = c\mathrm{d}y_t$ and $y = cy_t + d$, where $y_t \in \outdomain$.
The above equation is equivalent to
\begin{equation*}
    \begin{split}
        & Err(x', \mechanism') - Err(x', \mechanism'_{\text{bad}}) \\
        =& \int_{\outdomain} \loss(cy_t + d, cx+d)\left(\mathcal{P}_{\mechanism'(cx+d)} - \mathcal{P}_{\mechanism_{\text{bad}}'(cx+d)}\right)c\mathrm{d}y_t \\
        =& \int_{\outdomain} \loss(cy_t + d, cx+d)\frac{1}{c}\left(\mathcal{P}_{\mechanism(x)} - \mathcal{P}_{\mechanism_{\text{bad}}(x)}\right)c\mathrm{d}y_t.
    \end{split}
\end{equation*}
The last equality holds due to the definition of $\trans: \mechanism \to \mechanism'$.
For $L_p$-similar error metric $\loss$ (i.e. $\loss(y,x) \coloneq |y-x|^p$), it follows that
\begin{equation*}
    \loss(cy_t + d, cx+d) = \loss(cy_t, cx) = c^p\loss(y_t, x).
\end{equation*}
Thus, the above difference of $Err$ is equivalent to
\begin{equation*}
    \begin{split}
        & Err(x', \mechanism') - Err(x', \mechanism'_{\text{bad}}) \\
        =& c^p\int_{\outdomain} \loss(y_t, x)\left(\mechanism(x) - \mechanism_{\text{bad}}(x)\right)\mathrm{d}y_t \\
        =& c^p\left(Err(x, \mechanism) - Err(x, \mechanism_{\text{bad}})\right) \leq 0,
    \end{split}
\end{equation*}
due to the known fact $c>0$ and $Err(x, \mechanism) - Err(x, \mechanism_{\text{bad}}) \leq 0$.
\end{proof}

\emph{Remark:} Intuitively, this theorem is to prove: if $\mechanism$ is a better mechanism than $\mechanism_{\text{bad}}$ on $\indomain$,
then it is still a better mechanism than $\mechanism_{\text{bad}}$ after linearly mapping their outputs to $\indomain'$.

\subsection{Proof of Theorem~\ref{theo:optimal_concretization}} \label{appendix:optimal_concretization}

\begin{proof}
    Appendix~\ref{appendix:deduction} provides the formalized procedure. 
    Following this procedure, we show the optimal GPM under $\indomain \to \outdomain = [0,1) \to [0,1)$ and $\loss(y,x) = |y-x|$.
    
    The variables in TPM are $p$, $l$, and $r$. Since it is a family of probability distributions, the normalization constraint is
    \begin{equation*}
        (r - l) \cdot p + (1 - (r-l)) \cdot p/\exp(\varepsilon) = 1,
    \end{equation*}
    which means the length of the central piece is
    \begin{equation*}
        s \coloneq r - l = \frac{\exp(\varepsilon) - p}{p(\exp(\varepsilon) - 1)}.
    \end{equation*}
    Without loss of generality, assume $x = 0$ is the optimal point ($x = 1$ is symmetric).
    The optimization problem for solving the optimal $p$ is
    \begin{equation*}
        \begin{split}
            &\argmin_{p} \left( \int_{0}^{s} y\cdot p\thinspace\mathrm{d}y + \int_{s}^{1} y\cdot \frac{p}{\exp(\varepsilon)}\thinspace\mathrm{d}y \right) \\
            = &\argmin_{p} \left( \frac{s^2}{2}\left(p-\frac{p}{\exp(\varepsilon)}\right) + \frac{1}{2}\frac{p}{\exp(\varepsilon)} \right)\\
            = &\argmin_{p} \frac{1}{2}\left( \frac{(\exp(\varepsilon) - p)^2}{p(\exp(\varepsilon) - 1)\exp(\varepsilon)} + \frac{p}{\exp(\varepsilon)} \right).
        \end{split}
    \end{equation*}
    To solve the optimal $p$, we take the first-order derivative w.r.t. $p$ and set it to $0$, i.e.
    \begin{equation*}
        \frac{\partial}{\partial p} \left(\frac{(\exp(\varepsilon) - p)^2}{p(\exp(\varepsilon) - 1)\exp(\varepsilon)} + \frac{p}{\exp(\varepsilon)} \right) = 0,
    \end{equation*}
    This leads to $p = \exp(\varepsilon/2)$. Then
    \begin{equation*}
        s = \frac{\exp(\varepsilon/2) - 1}{\exp(\varepsilon) - 1}.
    \end{equation*}
    Having solved $p$ and $s$, the optimal $r$ is $r = l + s$. Then the optimal $l$ is determined by 
    \begin{equation*}
        \begin{split}
            \argmin_{l} & \int_{0}^{l} (x-y)\cdot \frac{p}{\exp(\varepsilon)}\thinspace\mathrm{d}y + \int_{l}^{x} (x-y)\cdot p\thinspace\mathrm{d}y + \\
            & \int_{x}^{l+s} (y-x)\cdot p\thinspace\mathrm{d}y + \int_{l+s}^{1} (y-x)\cdot \frac{p}{\exp(\varepsilon)}\thinspace\mathrm{d}y
        \end{split}
    \end{equation*}
    This is a univariate optimization problem w.r.t. $l$. 
    Moreover, it is a two-order polynomial w.r.t. $l$ and can be solved by analyzing the first-order and second-order derivatives.
    The solution is 
    \begin{equation*}
        l = \frac{2x(p - pe^{-\varepsilon}) - s (p - pe^{-\varepsilon})}{2(p - pe^{-\varepsilon})} = x - \frac{s}{2}.
    \end{equation*}
    Note that when $x - s/2 < 0$, the above $l$ is outside the domain $[0,1)$. In this case, the optimal $l$ is $l = 0$.

    Relating the above deduction to Theorem~\ref{theo:optimal_concretization}, the term $s/2$ corresponds to $C$.
    Then the optimal $p$ is $\exp(\varepsilon/2)$, $l = x - C$, and $r = x + C$ when $x \in [C, 1-C)$, which completes the proof
    for $\loss(y,x) = |y-x|$.
    The proof for $\loss(y,x) = |y-x|^2$ is similar.
\end{proof}

\emph{Remark:} This proof is the same as finding the optimal $3$-piecewise distribution in domain $[0,1)$.
The source code of our framework provides the validation.

\begin{figure}[t]
    \centering
    \begin{subfigure}[b]{0.35\linewidth}
        \centering
        \includegraphics[width=0.98\linewidth]{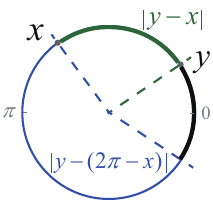}
        \caption{$x = 3\pi / 4$.}
    \end{subfigure}
    \hspace{0.1\linewidth}
    \begin{subfigure}[b]{0.35\linewidth}
        \centering
        \includegraphics[width=0.98\linewidth]{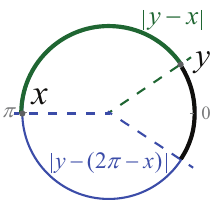}
        \caption{$x = \pi$.}
    \end{subfigure}
    \caption{Examples of $\lossmod(y, x)$ w.r.t. $x$. Given a specific $y$, the shorter between the blue and green arcs is $\lossmod(y, x)$.
    $\max_x\lossmod(y, x)$ is achieved at $x = \pi$.}
    \label{fig:appendix:circular_lemma}
\end{figure}

\subsection{Proof of Lemma~\ref{lemma:lossmod}} \label{appendix:lossmod}

\begin{proof}
    Note that $\lossmod(y, x) = \min \left(\loss(y,x), \loss(y, 2\pi - x)\right)$.
    The key observation is that for any $y$ and $L_p$-similar
    distance metric $\loss$, we have
    \begin{equation*}
        \max_x\lossmod(y, x) = \lossmod(y, \pi) = \loss(y, \pi).
    \end{equation*}

    Figure~\ref{fig:appendix:circular_lemma} illustrates the intuition.
    For any fixed $y$, it compares the length of $|y - x|$ and $|y - (2\pi - x)|$ w.r.t $x\in[0,2\pi)$.
    $\lossmod(y,x)$ is determined by the minimum of the two, and $\max_x\lossmod(y, x)$ will achieved at $x = \pi$.

    The remained proof is more intuitive than the proof of Lemma~\ref{lemma:endpoint},
    as the inner integrand $\lossmod(y, x)$ has a \emph{unique} maximum at $x = \pi$
    for any $y\in [l_i, r_i)$,
    making the swap of $\max_x$ and integration valid.
    Specifically, we have
    \begin{equation*}
        \begin{split}
            \max_x \sum_{i}^{m}p_i\int_{l_i}^{r_i}\lossmod(y,x)\mathrm{d}y &= \sum_{i}^{m}p_i \int_{l_i}^{r_i}\max_x \lossmod(y,x)\mathrm{d}y \\
            &= \sum_{i}^{m}p_i \int_{l_i}^{r_i} \loss(y,\pi)\mathrm{d}y
        \end{split}
    \end{equation*}
    holds trivially, because all other $x$ values always result in smaller $\lossmod(y,x)$.
    Therefore, for any values of $p_i\geq 0$, $l_i\leq r_i$,
    the integration of $\lossmod(y,x)$ is bounded by the integration of $\loss(y, \pi)$.
\end{proof}

\subsection{Proof of Theorem~\ref{theo:unbiased_optimal}} \label{appendix:unbiased_optimal}
\begin{proof}
    We prove the unbiasedness of $\mechanism$. The expectation of the given $\mechanism$ is
    \begin{equation*}
        \begin{split}
              & \mathrm{E}[\mechanism(x)] = \int_{-C}^{C+1} y \cdot \mathcal{P}_{\mechanism(x)} \mathrm{d}y                                                             \\
            = & \int_{-C}^{l}y\frac{p}{\exp(\varepsilon)}\mathrm{d}y + \int_{l}^{r}yp\mathrm{d}y + \int_{r}^{C+1}y\frac{p}{\exp(\varepsilon)}\mathrm{d}y \\
            = & \frac{r^2-l^2}{2}\left(p-\frac{p}{\exp(\varepsilon)}\right) + \left(C+\frac{1}{2}\right)\frac{p}{\exp(\varepsilon)}.
        \end{split}
    \end{equation*}
    Denote $s \coloneq (2C+1)(C-1)/(2C)$, which rewrites $l$ and $r$ as
    \begin{equation*}
        l = \frac{C+1}{2}x - \frac{C-1}{4} -\frac{s}{2},\  r = \frac{C+1}{2}x - \frac{C-1}{4} +\frac{s}{2}.
    \end{equation*}
    Then the above $\mathrm{E}[\mechanism(x)]$ is equivalent to
    \begin{equation*}
        \begin{split}
            &\frac{(C+1)sx - (C-1)s/2}{2}\cdot \frac{4C}{(C^2-1)(2C+1)} + \frac{\exp(-\varepsilon/2)}{2} \\
            =& x - \frac{C-1}{2(C+1)} + \frac{\exp(-\varepsilon/2)}{2} = x,
        \end{split}
    \end{equation*}
    leading to $\mathrm{E}[\mechanism] = x$, i.e. $\mechanism$ is unbiased.
\end{proof}

\section{Complementary Materials}

\subsection{Detailed Comparison with PTT (Section~\ref{subsec:piecewise})} \label{appendix:ptt}

Piecewise transformation technique (PTT)~\cite{ma10430379} is a framework for $3$-piecewise mechanisms.
It shows that (i) many PTT mechanisms are asymptotically optimal when used to obtain an unbiased estimator for mean of numerical data,
and (ii) there is a PTT that reaches the theoretical lower bound on variance.

Under the viewpoint of this paper, type-I PTT focuses on TPM that constrains the probabilities of the central interval ($p$) and the two side intervals ($q$) as
\begin{equation*}
        p = \frac{1}{2ak}\frac{\exp(\varepsilon)}{\exp(\varepsilon) - 1}, \quad q = \frac{\exp(\varepsilon)}{k(\exp(\varepsilon) - 1)},
\end{equation*}
where $a$ and $k$ are parameters to be determined. Type-II PTT focuses on TPM that constrains $p$ and $q$ as
\begin{equation*}
    p = \frac{1}{ak}\frac{\exp(\varepsilon)}{\exp(\varepsilon) - 1}, \quad q = \frac{\exp(\varepsilon)}{k(\exp(\varepsilon) - 1)}.
\end{equation*}
Therefore, PTT is still a specific case of the TPM framework.
Additionally, when discussing optimality of PTT, it gives a value of $a$ for type-I PTT but do not provide the optimal $k$.

\subsection{Related Optimality (Section~\ref{anchor:optimality_conditions})} \label{appendix:optimality_discussion}

In this paper, the optimality of GPM is defined with respect to:
(i) the worst-case $L_p$-similar error metric, 
(ii) bounded numerical domains $\indomain \to \outdomain$ and mechanisms based on piecewise distributions, 
(iii) minimization of error value (not asymptotic or order-of-magnitude optimality), and 
(iv) without post-processing.
$L_p$-similar error metrics are natural choices for evaluating data utility~\cite{DBLP:conf/stoc/HardtT10,DBLP:conf/isit/GengV14,DBLP:conf/icde/WangXYZHSS019}.
Bounded numerical domains are common in real-world applications.
Focusing on error values allows for more precise comparisons between different mechanisms.
By excluding post-processing, we can analyze the optimality of the mechanism itself, 
which provides a more fundamental understanding than considering the mechanism combined with a specific post-processing.

Other types of optimality have been explored in the literature, particularly for variants of Laplace mechanisms. 
The staircase mechanism~\cite{DBLP:conf/isit/GengV14} adopts the same utility model without prior knowledge or post-processing as this paper. 
It claims optimality under specific assumptions, one of which is that a staircase (piecewise) distribution \emph{can} achieve the optimal error. 
The mechanism demonstrates better $L_1$-error performance than the Laplace mechanism on $\outdomain = (-\infty, \infty)$, 
and its asymptotic optimality has been formally proven.
\ \emph{Universal optimality} is another type of optimality, defined from the perspective of a user's prior knowledge and post-processing ability~\cite{DBLP:conf/stoc/GhoshRS09}.
In this utility model, the user observes the output of the mechanism and selects another value based on the output and their prior knowledge,
i.e. under a Bayesian utility framework.
Formally, if the user's prior is denoted as $p_i$ on the data domain $i \in N$ (i.e. a discrete domain) and the user's post-processing is represented as a remap $z_{i,j}$ that reinterprets 
the output of the mechanism (on the sensitive value $i$) to $j$, then the utility model is defined as
\begin{equation*}
    Err(i) = \sum_{i\in N} p_i \sum_{j\in N} z_{i,j} \cdot \loss(i, j).
\end{equation*}
This utility model incorporates the user's prior knowledge and post-processing ability.
A mechanism is called universally optimal if, for any prior $p_i$, there exists an optimal remap $z_{i,j}$.
Under this utility model, it was proven that the truncated geometric mechanism (a discretized version of the Laplace mechanism)
can achieve universal optimality for count queries\footnote{
    This is in the centralized DP setting, where the data curator holds the dataset and uses \emph{one} mechanism.
}
and a legal error metric $\loss(i,j)$.
Such universal optimality was shown to be unachievable for more complex queries~\cite{DBLP:journals/siamcomp/BrennerN14}.
Under the same utility model, the universal optimality was extended
to the truncated Laplace mechanism for a bounded numerical domain $\indomain = [0,1]$ 
by approximating the geometric mechanism with the Laplace mechanism and post-processing~\cite{DBLP:conf/lics/FernandesMM21}.

These results do not hold in our utility model, i.e. utility model without prior and post-processing.
Figure~\ref{fig:exp:staircase_truncated} has shown that OGPM generally has a smaller error than the truncated Laplace mechanism,
especially when the privacy parameter $\varepsilon$ is not small,
indicating the sub-optimality of the truncated Laplace mechanism in the absence of using prior and post-processing.

\subsection{Directions for Analytically Proving Optimal $m = 3$ (Section~\ref{subsec:closed_form})} \label{appendix:one_direction}

This appendix outlines two potential directions for analytically proving that the optimal $m$ is $3$, along with the challenges associated with each approach.

Mathematically, finding the optimal $m$-piecewise mechanism is equivalent to identifying the optimal $m$-piecewise distribution under an $L_p$-similar error metric. 
It is seemingly true that the optimal $m$ is $3$: 
if the optimal $m$-piecewise distribution is not $3$ but $4$ or more, 
we can always shift the probability mass from the two side intervals (i.e. other pieces) to the central interval, thereby reducing the error. 
At the very least, the following fact holds:

\begin{fact}
    The optimal $m$-piecewise distribution has a strict staircase shape, i.e. the probability density of the central interval is greater than that of the two side intervals.
    \label{fact:optimal_m}
\end{fact}

\begin{figure}
    \centering
    \begin{subfigure}[b]{0.465\linewidth}
        \includegraphics[width=0.98\linewidth]{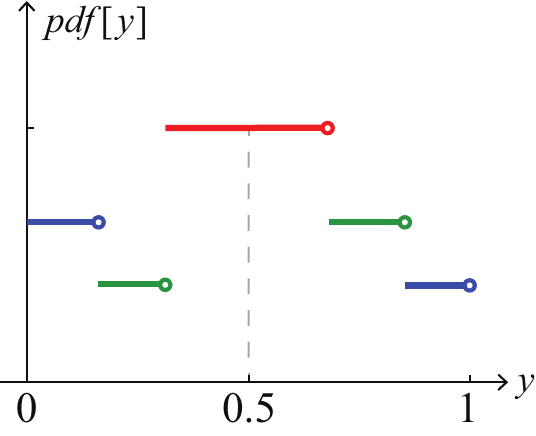}
        \caption{Non-staircase distribution.}
    \end{subfigure}
    \hfill
    \begin{subfigure}[b]{0.465\linewidth}
        \includegraphics[width=0.98\linewidth]{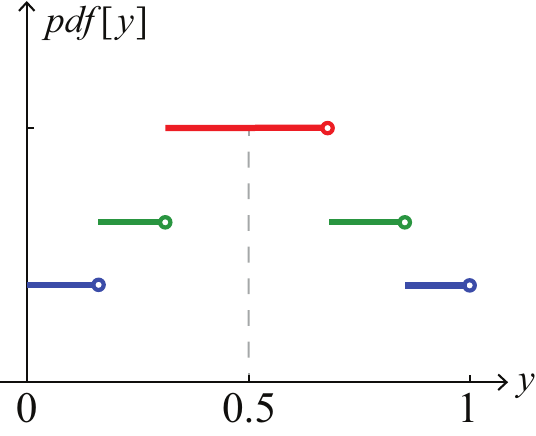}
        \caption{Staircase distribution.}
        \label{fig:appendix:direction_2}
    \end{subfigure}
    \caption{A non-staircase distribution (left) can always be shifted into a staircase distribution (right) by moving 
    some pieces closer to $x$, which reduces the error.}
    \label{fig:appendix:direction}
\end{figure}

Figure~\ref{fig:appendix:direction} illustrates this fact. Moving pieces while keeping their probabilities unchanged clearly maintains 
both the $\varepsilon$-LDP constraint and the probability normalization constraint. 
This observation reduces the problem to proving that a $3$-staircase distribution can achieve 
the same optimal error as a $4$-staircase distribution under the $\varepsilon$-LDP and probability normalization constraints.

\textbf{Direction 1:} If we can further move the green piece in Figure~\ref{fig:appendix:direction_2} ``into'' the red central piece while keeping 
the probabilities of the red and blue pieces unchanged, i.e. transform it into a $3$-staircase distribution while ensuring a decrease in the error,
then we can prove that the optimal $m$ is $3$.
However, this is challenging, as it breaks the probability normalization constraint (i.e. the sum of probabilities is no longer $1$),
requiring adjustments to the probabilities of each piece to satisfy the $\varepsilon$-LDP constraint.
The difficulty lies in ensuring that these adjustments will indeed decrease the error.

\textbf{Direction 2:} Another approach is to formulate the problems for $3$-staircase and $4$-staircase distributions as two constrained optimization problems. 
The goal would be to prove that the optimal error of the $3$-staircase distribution is equivalent to that of the $4$-staircase distribution. 
Ideally, these two multi-variable optimization problems could be solved analytically,
resulting in two closed-form error expressions w.r.t. $x$ and $\varepsilon$, thereby completing the proof for any $x$ and $\varepsilon$ by showing that the two expressions are equal.
This direction aligns with our framework. 
However, the challenge lies in the complexity of solving such multi-variable optimization problems analytically. 
This is why we rely on an off-the-shelf optimization solver, which, while effective, 
only provides numerical solutions for specific $x$ and $\varepsilon$ values.

\subsection{Analytical Deduction for the Optimal GPM (Section~\ref{subsec:closed_form})} \label{appendix:deduction}

If the optimal GPM under $\loss$ is proved to be TPM, then the closed-form optimal can be derived by deduction.

Denote $\indomain=[a,b)$ and $\outdomain = [\tilde{a},\tilde{b})$. 
Notice that the normalization constraint of probability is
\begin{equation*}
    (r - l) \cdot p + [(b - a) - (r - l)] \cdot p/\exp(\varepsilon) = 1.
\end{equation*}
This means the central interval length is
\vspace*{-0.3em}
\begin{equation*}
    s \coloneq r - l = \frac{1}{p(1-\exp(-\varepsilon))} - \frac{\tilde{b}-\tilde{a}}{\exp(\varepsilon)-1}.
\end{equation*}
If the minimal worst-case error is achieved at $x = a$ in Lemma~\ref{lemma:endpoint}, then
solving for the optimal $p$ is reduced to
\begin{equation*}
    \argmin_{p} \left( \int_{\tilde{a}}^{s} \loss(y, a)p\thinspace\mathrm{d}y + \int_{s}^{\tilde{b}} \loss(y, a)\frac{p}{\exp(\varepsilon)}\mathrm{d}y \right),
\end{equation*}
which is a univariate optimization problem w.r.t. $p$ and can be solved analytically.
With the solved $p$, Formulation~(\ref{equ:lr_i}) is also reduced to a univariate optimization
problem w.r.t. $l$:
\begin{equation*}
    \argmin_{l} \left(\int_{\tilde{a}}^{l} P_1 \thinspace \mathrm{d}y + \int_{l}^{l+s} P_2 \thinspace\mathrm{d}y + \int_{l+s}^{\tilde{b}} P_1 \thinspace \mathrm{d}y \right),
\end{equation*}
where $P_1 = \loss(y,x)p$ and $P_2 = \loss(y,x)p/\exp(\varepsilon)$.
This univariate optimization problem solves the optimal $l$, and the optimal $r$ is $r = l + s$.
Note that $l$ and $r$ should be restricted in $[\tilde{a},\tilde{b})$ when analyzing the first-order derivative.

\subsection{MSE of the Optimal GPM (Section~\ref{anchor:classical_mse})} \label{appendix:classical_mse}

Denote $p_\varepsilon$ and $C$ as the same as the instantiations of $\mechanism$ in Theorem~\ref{theo:optimal_concretization}.
The MSE of $\mechanism$ is
\begin{enumerate}[leftmargin=*, labelindent=0.1pt]
    \item If $x \in [0,C)$:
    \begin{equation*}
        \frac{p_\varepsilon}{3} \left((2C-x)^3 + x^3\right) + \frac{p_\varepsilon}{3\exp(\varepsilon)} \left((1-x)^3 - (2C-x)^3\right).
    \end{equation*}
    \item If $x \in [C,1-C)$:
    \begin{equation*}
        \frac{p_\varepsilon}{3\exp(\varepsilon)} \left(-2C^3 + 3x^2 -3x +1\right) +  \frac{p_\varepsilon}{3} \left(2C^3\right).
    \end{equation*}
    \item If $x \in [1-C,1)$:
    \begin{equation*}
        \frac{p_\varepsilon}{3\exp(\varepsilon)} \left((1-2C-x)^3 + x^3\right) + \frac{p_\varepsilon}{3} \left((1-x)^3 - (1-2C-x)^3\right).
    \end{equation*}
\end{enumerate}
For example, when $x = 0$, the MSE of $\mechanism$ is
\begin{equation*}
    \mathrm{MSE}[\mechanism(0)] = \frac{p_\varepsilon}{3} \left(8C^3\right) + \frac{p_\varepsilon}{3\exp(\varepsilon)} \left(1 - 8C^3\right).
\end{equation*}
Setting $\varepsilon = 1$ results in $\mathrm{MSE}[\mechanism(0)] = 0.22$ of OGPM.
As a comparison, SW~\cite{DBLP:conf/sigmod/Li0LLS20}, which also designed for $\indomain=[0,1)$, 
has an MSE of $0.29$ at $x = 0$.

\begin{figure}
    \centering
    \begin{subfigure}[b]{0.465\linewidth}
        \includegraphics[width=0.98\linewidth]{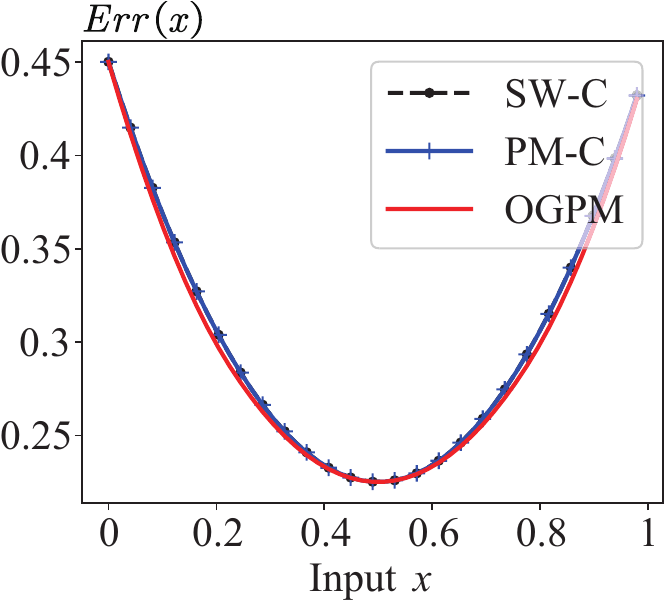}
        \caption{Privacy parameter $\varepsilon = 0.4$.}
    \end{subfigure}
    \hfill
    \begin{subfigure}[b]{0.465\linewidth}
        \includegraphics[width=0.98\linewidth]{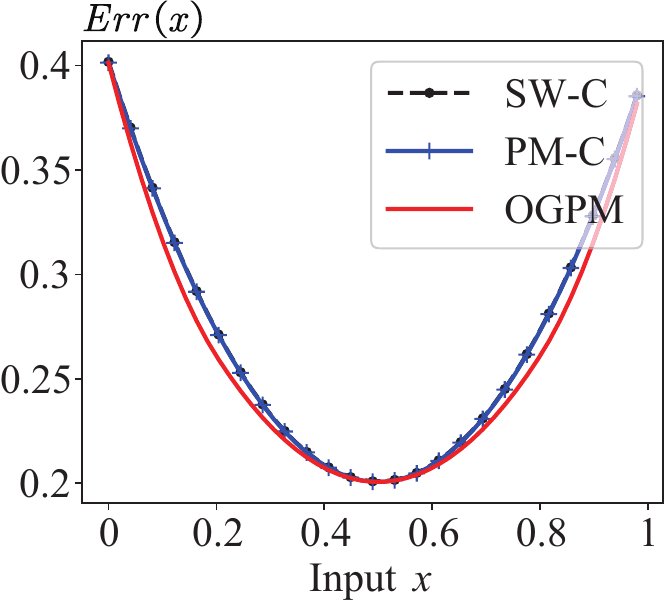}
        \caption{Privacy parameter $\varepsilon = 0.8$.}
    \end{subfigure}
    \caption{Whole-domain error comparison in the classical domain with error metric $\loss = |y-x|$.}
    \label{fig:appendix:small_epsilon}
\end{figure}

\subsection{Optimal Assignment of Privacy Parameter and an Example (Section~\ref{sec:polar})} \label{appendix:polar_assignment_epsilon}

The objective function to minimize the given total error in 2D polar coordinates is
\begin{equation*}
    \min_{\varepsilon_1, \varepsilon_2} Err_{\text{1,wor}}(\varepsilon_1) + Err_{\text{2,wor}}(\varepsilon_2),
\end{equation*}
where $Err_{\text{1,wor}}(\varepsilon_1)$ and $Err_{\text{2,wor}}(\varepsilon_2)$ are the worst-case errors of the classical domain and the circular domain, respectively.

Without loss of generality, we can assume the polar coordinate data is in $[0, 1) \times [0, 2\pi)$ and $\loss = |y_1-x_1|^2$, 
then $Err_{\text{1,wor}}(\varepsilon_1)$ and $Err_{\text{2,wor}}(\varepsilon_2)$ are already given by our MSE analysis.
However, the above optimization problem as a function of $\varepsilon_1$ and $\varepsilon_2$ 
is generally non-linear, thus hard to be analytically solved.
Therefore, a simple and practical way to find the optimal $\varepsilon_1$ and $\varepsilon_2$ is numerical testing.

Under distance metric $\loss(y,x) = |y-x|^2$, the worst-case error of
\begin{wrapfigure}[9]{r}{3.2cm}
    \centering
    \includegraphics[width=0.99\linewidth]{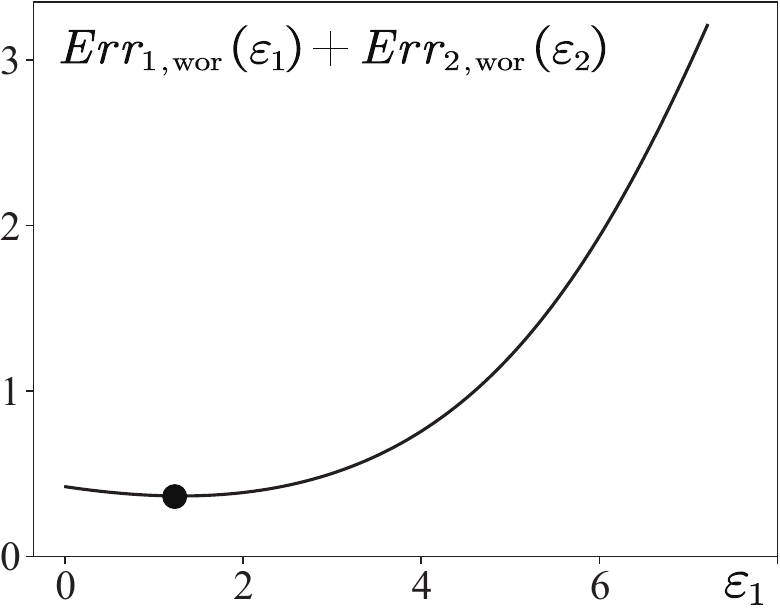}
    \label{fig:circular_lemma}
\end{wrapfigure}
the classical domain $[0,1)$ is achieved at $x = 0$. Therefore, $Err_{\text{1,wor}}$ equals to $\mathrm{MSE}[\mechanism(0)]$ calculated before.
For the circular domain $[0,2\pi)$, the worst-case error $Err_{\text{2,wor}}$ is stated in Theorem~\ref{theo:circular_mse}.
\ If $\varepsilon = 1 + 2\pi$ and we assign $\varepsilon_1$ to the classical domain and $\varepsilon_2 = \varepsilon-\varepsilon_1$ to the circular domain, 
then the total error is plotted in the right figure.
In this figure, the optimal assignment is $\varepsilon_1 = 1.32$ and $\varepsilon_2 = 5.69$.

From the curve of the total error, we can see that $\varepsilon_2$ affects the total error more than $\varepsilon_1$.
Even if $\varepsilon_1$ is set to $0$, the total error is not significantly affected, and it is
still  is significantly smaller than the case of $\varepsilon_2 = 0$.
This is because the circular domain has a larger range than the classical domain, 
thus the error of the circular domain is more sensitive to the privacy parameter.

Note that the optimality for the polar coordinate data is under the specific error metric $\loss_{\text{2D}} \coloneq \loss(y_1, x_1) + \lossmod(y_2, x_2)$.
If the error metric differs, the optimal error might not be preserved.

\subsection{Comparison under Small $\varepsilon$ (Section~\ref{subsec:comparison_with_PM_SW})} \label{appendix:smaller_epsilon}

Figure~\ref{fig:appendix:small_epsilon} presents the whole-domain error comparison of OGPM, PM-C, and SW-C under smaller $\varepsilon$ values, 
specifically $\varepsilon = 0.4$ and $\varepsilon = 0.8$. 
In these scenarios, all three mechanisms approach the uniform distribution more closely compared to cases with larger $\varepsilon$. 
Consequently, their errors are also more similar to each other. 
Statistically, when $\varepsilon = 0.4$, the error of OGPM is at most $0.008$ smaller than that of PM-C and SW-C. 
For $\varepsilon = 0.8$, the error of OGPM is at most $0.015$ smaller than that of PM-C and SW-C.

\subsection{Expected Error of the B-Laplace Mechanism (Section~\ref{subsec:staircase_truncated})} \label{appendix:b_laplace}

The B-Laplace mechanism redefines a Laplace-shaped distribution on a bounded domain as the perturbation mechanism. 
For the data domain $\indomain \to \outdomain = [0,1) \to [0,1)$, the B-Laplace mechanism is defined as follows:

\begin{definition}[Bounded Laplace Mechanism, adapted from~\cite{DBLP:journals/jpc/HolohanABA20}]
    The B-Laplace mechanism $\mechanism(x):[0,1) \to [0,1)$ is given by the probability density function (PDF) as follows:
    \begin{equation*}
        pdf[\mathcal{M}(x) = y] = \frac{1}{C_y}\cdot \frac{1}{2b}\exp\left(-\frac{|y-x|}{b}\right) \quad \forall y \in [0,1),
    \end{equation*}
    where $b$ is the scale parameter, and $C_y = \int_{0}^{1} \frac{1}{2b}\exp\left(-\frac{|y-x|}{b}\right) \mathrm{d}x$ is the normalization constant.
\end{definition}
According to Theorem 3.5 and Corollary 4.5 in~\cite{DBLP:journals/jpc/HolohanABA20}, the B-Laplace mechanism satisfies $\varepsilon$-LDP whenever $b \geq 1/\varepsilon$.
Using the best scale parameter $b = 1/\varepsilon$, the normalization constant becomes $C_y = (1 - \exp(-\varepsilon)) / 2$.
We can compute the expected $L_1$ error of the B-Laplace mechanism as follows (this computation is not included in~\cite{DBLP:journals/jpc/HolohanABA20}):
\begin{equation*}
    \begin{split}
        Err(x, \mechanism) &= \int_{0}^{1} |y-x| \cdot pdf[\mathcal{M}(x) = y] \mathrm{d}y \\
        &= \int_{0}^{1} |y-x| \cdot \frac{1}{C_y}\cdot \frac{1}{2b}\exp\left(-\frac{|y-x|}{b}\right) \mathrm{d}y \\
        &= \frac{\varepsilon}{1 - \exp(-\varepsilon)} \int_{0}^{1} |y-x| \exp\left(-\varepsilon|y-x|\right) \mathrm{d}y.
    \end{split}
\end{equation*}
The above integral can be numerically computed using the Python library function \verb|scipy.stats.laplace.expect()| or analytically solved. 
The final result for the expected error is
\begin{equation*}
    \frac{2 - (1 + \varepsilon x)e^{-\varepsilon x} - (1 + \varepsilon (1 - x)) e^{-\varepsilon(1-x)}}{\varepsilon(1 - e^{-\varepsilon})},
\end{equation*}
which is a closed-form expression w.r.t. $x$ and $\varepsilon$.

\end{document}